\documentclass[pra,twocolumn,floatfix,superscriptaddress,notitlepage,longbibliography]{revtex4-1}

\usepackage[normalem]{ulem}
\usepackage{graphicx, color, graphpap}
\usepackage{enumitem}
\usepackage{amssymb}
\usepackage{amsthm}
\usepackage{multirow}
\usepackage[colorlinks=true,citecolor=magenta,linkcolor=magenta]{hyperref}
\usepackage[T1]{fontenc}
\usepackage{verbatim}
\usepackage{mathtools}
\usepackage{titlesec}
\usepackage{float}

\long\def\ca#1\cb{} 

\newcommand{\braket}[2]{\langle #1 \hspace{1pt} | \hspace{1pt} #2 \rangle}

\newcommand{\ketbra}[2]{| \hspace{1pt} #1 \rangle \langle #2 \hspace{1pt} |}

\newcommand{\bramatket}[3]{\langle #1 \hspace{1pt} | #2 | \hspace{1pt} #3 \rangle}

\newcommand{\norm}[2][]{#1| \! #1| #2 #1| \! #1|}

\newcommand{\ket}[1]{|#1\rangle}               
\newcommand{\bra}[1]{\langle #1|}              
\newcommand{\dya}[1]{\ket{#1}\!\bra{#1}}

\newcommand{\poly}{\operatorname{poly}}

\newcommand{\CC}{\mathcal{C}}

\newcommand{\HC}{\mathcal{H}}

\newcommand{\OC}{\mathcal{O}}
\newcommand{\PC}{\mathcal{P}}

\newcommand{\SC}{\mathcal{S}}

\newcommand{\ZC}{\mathcal{Z}}

\newcommand{\Tr}{{\rm Tr}}

\renewcommand{\geq}{\geqslant}
\renewcommand{\leq}{\leqslant}

\renewcommand{\vec}[1]{\boldsymbol{#1}}  

\newcommand*{\id}{\openone}

{}
{}
\newtheorem{theorem}{Theorem}
\newtheorem{lemma}{Lemma}
\newtheorem{proposition}{Proposition}

\newtheorem{definition}{Definition}

\begin{document}
\title{Computable and operationally meaningful multipartite entanglement measures}

\author{Jacob L. Beckey}
\affiliation{Theoretical Division, Los Alamos National Laboratory, Los Alamos, New Mexico 87545, USA}
\affiliation{JILA, NIST and University of Colorado, Boulder, Colorado 80309, USA}
\affiliation{Department of Physics, University of Colorado, Boulder, Colorado 80309, USA}
\affiliation{Quantum Science Center, Oak Ridge, TN 37931, USA}

\author{N. Gigena}
\affiliation{Faculty of Physics, University of Warsaw, Pasteura 5, 02-093 Warsaw, Poland}

\author{Patrick J. Coles}
\affiliation{Theoretical Division, Los Alamos National Laboratory, Los Alamos, New Mexico 87545, USA}
\affiliation{Quantum Science Center, Oak Ridge, TN 37931, USA}

\author{M. Cerezo}
\affiliation{Theoretical Division, Los Alamos National Laboratory, Los Alamos, New Mexico 87545, USA}
\affiliation{Quantum Science Center, Oak Ridge, TN 37931, USA}
\affiliation{Center for Nonlinear Studies, Los Alamos National Laboratory, Los Alamos, New Mexico 87545, USA}

\begin{abstract} 
Multipartite entanglement is an essential resource for quantum communication, quantum computing, quantum sensing, and quantum networks. The utility of a quantum state, $\ket{\psi}$, for these applications is often directly related to the degree or type of entanglement present in $\ket{\psi}$. Therefore, efficiently quantifying and characterizing multipartite entanglement is of paramount importance. In this work, we introduce a family of multipartite entanglement measures, called Concentratable Entanglements. Several well-known entanglement measures are recovered as special cases of our family of measures, and hence we provide a general framework for quantifying multipartite entanglement. We prove that the entire family does not increase, on average, under Local Operations and Classical Communications. We also provide an operational meaning for these measures in terms of probabilistic concentration of entanglement into Bell pairs. Finally, we show that these quantities can be efficiently estimated on a quantum computer by implementing a parallelized SWAP test, opening up a research direction for measuring multipartite entanglement on quantum devices. 

\end{abstract}
\maketitle

\textit{Introduction.} The presence of entanglement in quantum states is widely recognized as one of, if not the, defining property of quantum mechanics~\cite{einstein1935can}. Since the development of quantum information theory~\cite{wilde2013quantum,plenio1998entanglement} it was realized that entanglement is a fundamental resource~\cite{horodecki2009quantum,gigena2020one} for quantum communications~\cite{bennett1993teleporting,barrett2002nonsequential,cleve1997substituting,gigena2017bipartite}, quantum cryptography~\cite{ekert1991quantum,gisin2002quantum}, and quantum computing~\cite{ekert1998quantum,nielsen2000quantum,datta2005entanglement}. Recent advances in quantum control technologies have made it possible to harness the power of entanglement for  quantum-enhanced sensing~\cite{chalopin2018quantum,pooser2020truncated,pezze2021entanglement} and communications~\cite{liao2017satellite,ma2012quantum,valivarthi2020teleportation}, and for showing quantum advantage using near-term quantum computers~\cite{google2019supremacy}. While the ubiquity of entangled quantum states as a resource is clear, their utility  for these applications often depends on the degree of entanglement in the quantum state. 

The nature of quantum entanglement is well understood for bipartite pure quantum states~\cite{zyczkowski2006geometry,nielsen2000quantum,wilde2013quantum}. However, the same cannot be said for the multipartite entanglement of pure states~\cite{walter2016multipartite}, where the complexity of entanglement scales exponentially with the number of parties. In fact, already for a system of three qubits there  exists two different, and  inequivalent, types of genuine tripartite entanglement, such that states of the two different kinds cannot be exactly  transformed onto the other via the action of Local Operations and Classical Communications (LOCC)~\cite{dur2000three}.  While the study of multipartite entanglement has received considerable attention~\cite{coffman2000distributed,acin2000generalized,barnum2001monotones,wong2001potential,meyer2002global,miyake2003classification,toth2005detecting,walter2013entanglement,sawicki2012critical} there does not exist a single unambiguous way to detect, quantify and characterize multipartite entanglement. Hence, improving our knowledge on the nature of the entanglement  between multiple parties is not only crucial to better understanding the underlying structure of quantum mechanics, but it is also a fundamental step towards enhancing emergent technologies such as distributed quantum sensing~\cite{guo2020distributed}, longer baseline telescopes~\cite{gottesman2012longer}, and various quantum internet applications~\cite{kimble2008quantum,navascues2020genuine,wehner2018quantum}.

The advent of quantum computing technologies brings forth the possibility of verifying and characterizing the multipartite entanglement present in states prepared on these near-term quantum devices. In this context, entanglement measures that are not only theoretically relevant, but that can also be estimated via  quantum algorithms~\cite{cerezo2020variationalreview,toth2005detecting,alves2004multipartite,smith2017quantifying,foulds2020controlled,hayden2013two,gutoski2013quantum}, become particularly attractive as characterization tools. For instance,  it was shown~\cite{brennen2003observable,meyer2002global} that given an $n$-qubit state $\ket{\psi}$, the linear entropies $\frac{1}{2}\left(1-\Tr\rho_j^2\right)$ of the single-qubit reduced states $\rho_j$ can be used to study the entanglement in $\ket{\psi}$. Moreover, since the SWAP test~\cite{buhrman2001quantum,harrow2013testing,gutoski2015quantum,cincio2018learning,subacsi2019entanglement} can be used to compute linear entropies, these measures can be efficiently estimated on quantum computers or optical quantum devices.

In this work, we introduce a family of quantities, called \textit{Concentratable Entanglements}, which characterize and quantify the multipartite entanglement in an arbitrary $n$-qubit pure state $\ket{\psi}$. We first prove that each Concentratable Entanglement does not increase, on average, under LOCC operations, and hence forms an entanglement monotone. Then, we show that by combining Concentratable Entanglements one can obtain several quantities of interest, which can quantify properties such as the entanglement in, and between, subsystems, as well as the total entanglement in $\ket{\psi}$. We then discuss how these quantities can be efficiently estimated on a quantum computer given two copies of $\ket{\psi}$,  and employing constant-depth $n$-qubit parallelized SWAP tests. Finally, we discuss the operational meaning of the Concentratable Entanglement as the probability of obtaining Bell pairs between qubits in the different copies of~$\ket{\psi}$.

Our results generalize previous results in the literature in the sense that: (1) several entanglement measures correspond to a special case of the Concentratable Entanglements~\cite{brennen2003observable,meyer2002global,walter2013entanglement,wong2001potential,carvalho2004decoherence}, (2) we prove a conjecture in Ref.~\cite{foulds2020controlled}, where it was hypothesized that the parallelized SWAP test can provide the basis for constructing a pure state  multipartite entanglement monotone. Finally, the broader implication of our work is to promote a research direction of studying multipartite entanglement using quantum devices, such as cloud-based quantum computers.

\bigskip

\textit{Concentratable Entanglement.} Consider an $n$-qubit pure quantum state  $\ket{\psi}$. We denote $\SC=\{1, 2, \ldots, n\}$ as the set of labels for each qubit, and $\PC(\SC)$ as its power set (i.e., the set of subsets, with cardinality $|\SC| =2^n$). We introduce the \textit{Concentratable Entanglements} as a family of entanglement monotones that characterize and quantify the multipartite entanglement in  $\ket{\psi}$. 
\begin{definition} \label{def:C} 
For any set of qubit labels $s\in \PC(\SC)\setminus \{\emptyset\}$, the Concentratable Entanglement is defined as
\begin{equation}
\CC_{\ket{\psi}}(s)=1-\frac{1}{2^{c(s)}}\sum_{\alpha\in \PC(s)} {\rm Tr} \rho_\alpha^2 \label{eq:concentratable}\,,
\end{equation}
where $c(s)$ is the cardinality of the set  $s$, and $\PC(s)$ its power set. Here we denote by $\rho_\alpha$ the joint reduced state, associated to $\ket{\psi}$, of the subsystems labeled by the elements in $\alpha$ (with $\alpha=\emptyset$ leading to $\rho_\alpha := 1$).
\end{definition}

Equation~\eqref{eq:concentratable} shows that each $\CC_{\ket{\psi}}(s)$ is the average of the entanglement between the subsets of qubits  with labels in $s$ and the rest of the system. This means that different Concentratable Entanglements can measure both bipartite and multipartite entanglements according to how $s$ is defined. For instance taking the smallest set possible, i.e., $s=\{j\}$ with $j=1,\ldots,n$, one finds $\CC_{\ket{\psi}}(\{j\})= \frac{1}{2}\left(1-\Tr\rho_j^2\right)$. Thus, when averaged over $\{j\}$, one recovers the measures in~\cite{brennen2003observable,meyer2002global} which quantify the bipartite entanglement between the $j$-th qubit and the rest. On the other hand, taking the largest set possible, i.e., $s=\SC$, $\CC_{\ket{\psi}}(\SC)$ quantifies the overall entanglement in $\ket{\psi}$ across all cuts, and as discussed below, this case corresponds to the entanglement measure conjectured in~\cite{foulds2020controlled}. Moreover, in this case we also recover the entanglement measure of~\cite{carvalho2004decoherence} as a special case of the Concentratable Entanglements. Including these extremal cases, there are a total of $2^n-1$ Concentratable Entanglements according to Definition~\ref{def:C}.

\begin{figure}[t!]
\centering
\includegraphics[width=1\columnwidth]{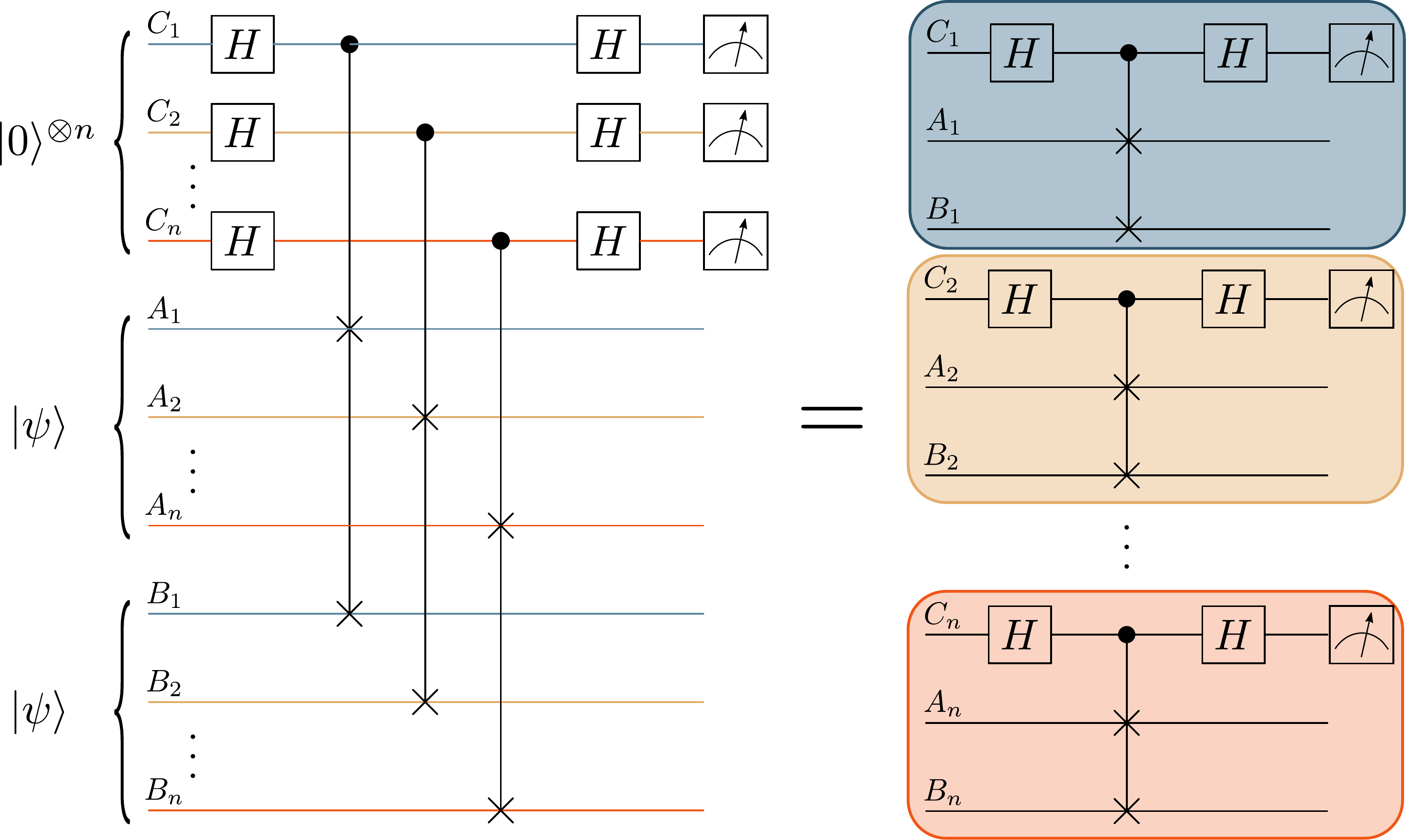}
\caption{\textbf{Circuit for the $n$-qubit parallelized SWAP test.} Given two copies of the quantum state $\ket{\psi}$,and $n$ ancilla qubits, the $n$-qubit parallelized SWAP test consists of employing  the $k$-th ancilla to perform a controlled swap test on the $k$-th qubit of each copy of $\ket{\psi}$. Since  the $n$  SWAP test can be factorized, one can perform them in parallel, leading to a constant depth circuit. }
\label{fig:SWAPn}
\end{figure}

\bigskip

\textit{Efficient Computation.} A fundamental aspect of the Concentratable Entanglements is that they can be efficiently estimated on a quantum computer. While each purity, $\Tr[\rho_{\alpha}^2]$, in~\eqref{eq:concentratable} can be computed via an overlap test~\cite{cincio2018learning}, one can also use two copies of the state $\ket{\psi}$ and $n$ ancilla qubits to  employ the $n$-qubit parallelized SWAP test depicted in Fig.~\ref{fig:SWAPn} (see Supplementary Information for a discussion on the SWAP test). From Fig.~\ref{fig:SWAPn}, it is clear that the $k$-th ancilla qubit is used to perform a controlled SWAP test on the $k$-th qubit of each copy of $\ket{\psi}$. The tests are independent and thus factorizable. This implies that the $n$-qubit parallelized SWAP test has a constant circuit depth for any number of qubits. 

Given the $n$-qubit parallelized SWAP test, we define the following relevant quantities. First, let $p(\vec{z})$ be the probability of measuring the $\vec{z}$ bitstring on the $n$ control qubits, and let $\ZC=\{0,1\}^n$ be  the set of all such bitstrings. Then, the following proposition (proved in the Supplementary Information~
\footnote{See Supplementary Information, which contains Refs.~\cite{nielsen2001majorization,nielsen1999conditions,vidal1999entanglement,bruss2019quantum,schatzki2021entangled,coles2019strong}
}) holds. 
\begin{proposition}\label{prop:CE-probab}
The Concentratable Entanglement can be computed from the outcomes of the $n$-qubit parallelized SWAP test as 
\begin{align}
\CC_{\ket{\psi}}(s)=1-\sum_{\vec{z}\in \ZC_0(s)}p(\vec{z})
\label{eq:concentratable2}\,,
\end{align}
where $\ZC_0(s)$ is the set of all bitstrings with $0$'s on all indices in $s$. 
\end{proposition}
Proposition~\ref{prop:CE-probab} shows that $\CC_{\ket{\psi}}(s)$ can be computed by performing the parallelized SWAP test on all qubits and adding the probabilities where the control qubits with indices in $s$ are measured in the $\ket{0}$ state. Since this corresponds to a conditional probability, one can also perform SWAP tests only on the qubits with indexes in $s$ (requiring just $c(s)$ ancillary qubits) and express the Concentratable Entanglement as $\CC_{\ket{\psi}}(s)=1-p(\vec{0}_x)$. Here $p(\vec{0}_x)=\sum_{\vec{z}\in \ZC_0(s)}p(\vec{z})$ denotes the probability of obtaining the all-zero result from the SWAP test on the qubits with labels in $s$. 

Here we remark that Eqs.~\eqref{eq:concentratable} and~\eqref{eq:concentratable2} are complimentary in the sense that the number of terms in the summations are inversely proportional. That is,  the summation in Eq.~\eqref{eq:concentratable} contains $2^{c(s)}$ terms, while that of~\eqref{eq:concentratable2} contains   $2^{n-c(s)}$ terms. Hence, we remark that it is preferable to employ Eq.~\eqref{eq:concentratable2} when analyzing multipartite entanglement, as this avoids potentially having to compute a prohibitively large number of purities as those required in other entanglement measures~\cite{scott2004multipartite}. For instance, if $(n-c(s)) \in\OC(\log(n))$ then~\eqref{eq:concentratable2} only contains number of terms in $\OC(\poly(n))$. For the purpose of analyzing multipartite entanglement we henceforth focus on Eq.~\eqref{eq:concentratable2}.

Finally, we remark, that, as shown in the Supplementary Information, the SWAP test still works when the two copies of $\ket{\psi}$ are not exactly the same. Specifically, let $\ket{\psi}$ and $\ket{\psi'}$ be two faulty copies of the state with  $\Vert \dya{\psi}-\dya{\psi'}\Vert_1 \leq \varepsilon$, then, we find that the error in the Concentratable Entanglement is upper-bounded by $\OC(\varepsilon^2)$, indicating that small errors in the state lead to small Concentratable Entanglement difference.

\bigskip

\textit{Properties of $\CC_{\ket{\psi}}(s)$.} We now present our main results which provide properties and additional insight for the Concentratable Entanglements. The proofs of these results are provided in the Supplementary Information.
\begin{theorem}\label{theo1}
The Concentratable Entanglement has the following properties:
\begin{enumerate}
    \item $\CC_{\ket{\psi}}(s)$ is non-increasing, on average, under LOCC operations and hence is a well-defined pure state entanglement measure. 
    \item If $\ket{\psi}$ is a separable state of the form $\ket{\psi}=\bigotimes_{j=1}^n \ket{\phi_j}$, then  $\CC_{\ket{\psi}}(s)=0$ for all $s\in\PC(\SC)\setminus \{\emptyset\}$.
    \item $\CC_{\ket{\psi}}(s')\leq \CC_{\ket{\psi}}(s) $ if $s'\subseteq s$.  
    \item Subadditivity,   $\CC_{\ket{\psi}}(s\cup s')\leq \CC_{\ket{\psi}}(s) +\CC_{\ket{\psi}}(s') $ for $s\cap s'=\emptyset$.  
    \item Continuity, let $\ket{\psi}$ and $\ket{\phi}$ be two states such that $\Vert \dya{\psi}-\dya{\phi}\Vert_1 \leq \varepsilon$, then $|\CC_{\ket{\psi}}(s)-\CC_{\ket{\phi}}(s)|\leq 2 \varepsilon$.
\end{enumerate}
\end{theorem}

Here, 3) guarantees that the Concentratable Entanglement always measures less entanglement in any subsystem of $s$. In addition, we remark that combining 3) and 4) we have  $\{\CC_{\ket{\psi}}(s) ,\CC_{\ket{\psi}}(s')\}\leq  \CC_{\ket{\psi}}(s\cup s') \leq \CC_{\ket{\psi}}(s) +\CC_{\ket{\psi}}(s')$ for $s\cap s'=\emptyset$.

To further understand how the Concentratable Entanglements measure entanglement, we provide additional details on the probabilities $p(\vec{z})$.  First, consider the following explicit formula for the probabilities $p(\vec{z})$.
\begin{proposition}
Given the expansion of the state   $\ket{\psi}=\sum_{\vec{i}} c_{\vec{i}}\ket{i_1}\ket{i_2}\cdots\ket{i_n}$, the probability  $p(\vec{z})$ for any $\vec{z}\in \ZC$ is given by
\begin{equation}
    p(\vec{z})=\frac{1}{2^n}\sum_{\substack{\vec{i},\vec{i'},\vec{j},\vec{j'}}} c_{\vec{i}} c_{\vec{i'}} c^*_{\vec{j}} c^*_{\vec{j'}} T_{\vec{i}\vec{i'}\vec{j}\vec{j'}}(\vec{z})\,,
\end{equation}
where $T_{\vec{i}\vec{i'}\vec{j}\vec{j'}}(\vec{z})=\prod_k(\delta_{i_k j_k}\delta_{i_k' j'_k}+(-1)^{z_k}\delta_{i_k j'_k}\delta_{i_k' j_k})$, 
and where $z_k$ denotes the $k$-th bit in $\vec{z}$.
\end{proposition}

Alternatively, one can also express $p(\vec{z})$ as a function of purities of reduced states of $\ket{\psi}$. Let us define as $w(\vec{z})$ the Hamming weight of $\vec{z}$, and let $\SC_1 \subseteq \SC$ be the set of labels for the bits in $\vec{z}$ that are equal to $1$ (with $|\SC_1|=w(\vec{z})$). Finally, let $c_{hs}$ be the cardinality of $\SC_h\cap s$. One finds  
\begin{equation}\label{eq:prob-purities}
    p(\vec{z})=\frac{1}{2^n}\sum_{s\in\PC(\SC)} (-1)^{c_{hs}} {\rm Tr}\,\rho_x^2\,.
\end{equation}

Equation~\eqref{eq:prob-purities} leads to the following proposition.
\begin{proposition}\label{prop:odd}
If  $\vec{z}$ has odd Hamming weight (if $w(\vec{z})$ is odd), then $p(\vec{z})=0$.
\end{proposition}
Proposition~\ref{prop:odd} has several implications. First, one can see that by performing the $n$-qubit parallelized SWAP test, one can never measure a bitstring with an odd number of ones. Then, the formula for the Concentratable Entanglements in Proposition~\ref{prop:CE-probab} can be expressed as
\begin{align}
\CC_{\ket{\psi}}(s)=\sum_{\vec{z}\in \ZC^{\text{even}}_1(s)}p(\vec{z})
\label{eq:concentratableven}\,,
\end{align}
where we recall that $\ZC_0(s)$ was defined as the set of all bitstrings with $0$'s on all indices in $s$, and where we respectively define $\ZC^{\text{even}}_1(s)$ and $\ZC^{\text{odd}}_1(s)$ as the compliments of $\ZC_0(s)$ with even and odd Hamming weight,  such that $\ZC_0(s)\cup \ZC^{\text{even}}_1(s)\cup \ZC^{\text{odd}}_1(s)=\ZC$. Simply said,  $Z^{\text{even}}_1(s)$ is the set of bitstrings with even Hamming weight and with at least a $1$ in an index in $s$. For instance, if  $s=\SC$ (i.e., when the Concentratable Entanglement measures all the correlations in $\ket{\psi}$) then $\CC_{\ket{\psi}}(\SC)=1-p(\vec{0})=\sum_{\vec{z}: w(\vec{z}) \text{ even}}p(\vec{z})$, and we recover exactly the conjectured measure of entanglement of~\cite{foulds2020controlled}.

Equation~\eqref{eq:concentratableven} shows that the information of the multipartite entanglement in $\ket{\psi}$ is encoded in the probabilistic outcomes of the $n$-qubit parallelized SWAP test when an even number of control qubits are measured in the $\ket{1}$ state. For instance, the probability of measuring a bitstring with Hamming weight $w(\vec{z})=2$, where $z_k=z_{k'}=1$ contains information regarding the bipartite entanglement between qubits $k$ and $k'$. Specifically, the following proposition holds.
\begin{proposition}\label{prop:bi-separable}
Let $\ket{\psi}$ be a bi-separable state $\ket{\psi}=\ket{\psi}_A\otimes\ket{\psi}_B$. Then for any bitstring $\vec{z}$ of Hamming weight $w(\vec{z})=2$, where $z_k=z_{k'}=1$ we have  $p(\vec{z})=0$ if qubit $k$ is in subsystem $A$, and qubit $k'$ is in subsystem $B$.
\end{proposition}
Proposition~\ref{prop:bi-separable} can be generalized to show that the probability of measuring a bitstring with Hamming weight $w(\vec{z})$ contains information regarding the entanglement between the qubits with labels in $\SC_1$. That is, one can prove  that $p(\vec{z})$ is equal to zero if the qubits in $\SC_1$ belong to non-entangled partitions of $\ket{\psi}$.

Here we remark that while the $p(\vec{z})$ contain information regarding the multipartite entanglement in $\ket{\psi}$, these probabilities are generally not entanglement monotones. The exception being $p(\vec{1})$ when $n$ is even, i.e., the probability of measuring all the control qubits in the $\ket{1}$ state. For this special case we find the following.
\begin{proposition}\label{prop:tangle}
If $n$ is even, then $p(\vec{1})$ is an entanglement monotone. Moreover, in this case     $p(\vec{1})= \tau_{(n)}/2^n$.
where $\tau_{(n)}$ is the $n$-tangle.
\end{proposition}
The $n$-tangle was introduced in~\cite{wong2001potential} as a measure of multipartite entanglement in $n$-qubits states that generalizes the Concurrence~\cite{wootters1998entanglement,jaeger2003entanglement}. The $n$-tangle of a pure state $\ket{\psi}$ is  $    \tau_{(n)}=|\braket{\psi}{\widetilde{\psi}}|^2$, with $ \ket{\widetilde{\psi}}= \sigma_y ^{\otimes n } \ket{\psi^*}$ 
where $\sigma_y$ is the Pauli-$Y$ operator, and $\ket{\psi^*}$ is the conjugate of $\ket{\psi}$. Hence, for the special case of two-qubits one finds that $C_{\ket{\psi}}(s)=\tau_{(2)}/4=C^2/4$ for all $s\in\PC(\SC)$, where here $C$ denotes the Concurrence~\cite{wootters1998entanglement}. In general, we see from  Proposition~\ref{prop:tangle} that $n$-tangle is always one of the terms in the summation of Eq.~\eqref{eq:concentratableven}, and hence is included in the Concentratable Entanglements.

\begin{figure}[t!]
\centering
\includegraphics[width=0.95\columnwidth]{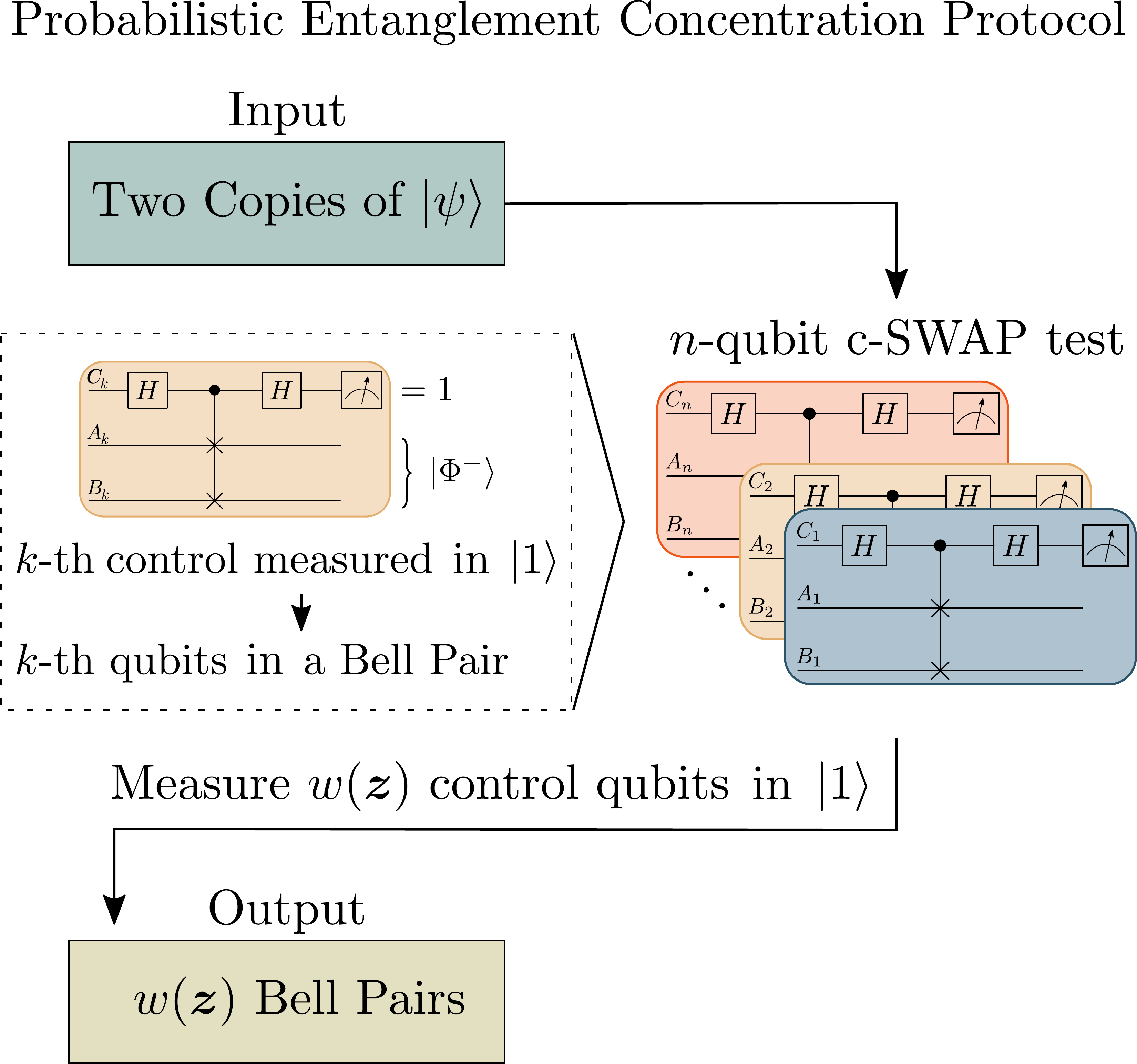}
\caption{\textbf{Protocol for concentrating entanglement.} Given two copies of  $\ket{\psi}$, one can employ the $n$-qubit parallelized SWAP test to prepare Bell pairs between qubits in the different copies of $\ket{\psi}$. Specifically, measuring the $k$-th control qubit in the state $\ket{1}$ implies that the joint state of the $k$-th qubit of each copy of $\ket{\psi}$ is the Bell state $\ket{\Phi^-}=\frac{1}{\sqrt{2}}\left(\ket{01}-\ket{10} \right)$. Hence, a single run of the $n$-qubit parallelized SWAP test has a probability $p(\vec{z})$ of concentrating the multipartite entanglement in the copies of $\ket{\psi}$ and producing $w(\vec{z})$ Bell pairs.  }
\label{fig:concentration}
\end{figure}

Interestingly, $p(\vec{z})$ can also be interpreted  as the probability of concentrating the entanglement in the two copies of $\ket{\psi}$ and ``distilling'' Bell pairs. 
\begin{proposition}\label{prop:BellPair}
Given two copies of $\ket{\psi}$,  if the $k$-th control qubit of the  $n$-qubit parallelized SWAP test  was measured in the state $\ket{1}$, then the joint post-measured state of the $k$-th qubits of each copy of $\ket{\psi}$ is the Bell state $\ket{\Phi^-}=\frac{1}{\sqrt{2}}\left(\ket{01}-\ket{10} \right)$.
\end{proposition}
Proposition~\ref{prop:BellPair} shows that when one measures (with probability $p(\vec{z})$) a bitstring $\vec{z}$ with (even) Hamming weight $w(\vec{z})$, then one has produced $w(\vec{z})$ Bell pairs between qubits in the different copies of $\ket{\psi}$ with indices in $\SC_h$. This protocol is schematically shown in Fig.~\ref{fig:concentration}. In addition, Proposition~\ref{prop:BellPair} also sheds additional light on the Concentratable Entanglement $\CC_{\ket{\psi}}(s)$ as the probability of obtaining any of the qubits pairs with labels in $s$ in a Bell pair when performing a SWAP test.

\bigskip

\begin{figure}[t!]
\centering
\includegraphics[width=0.95\columnwidth]{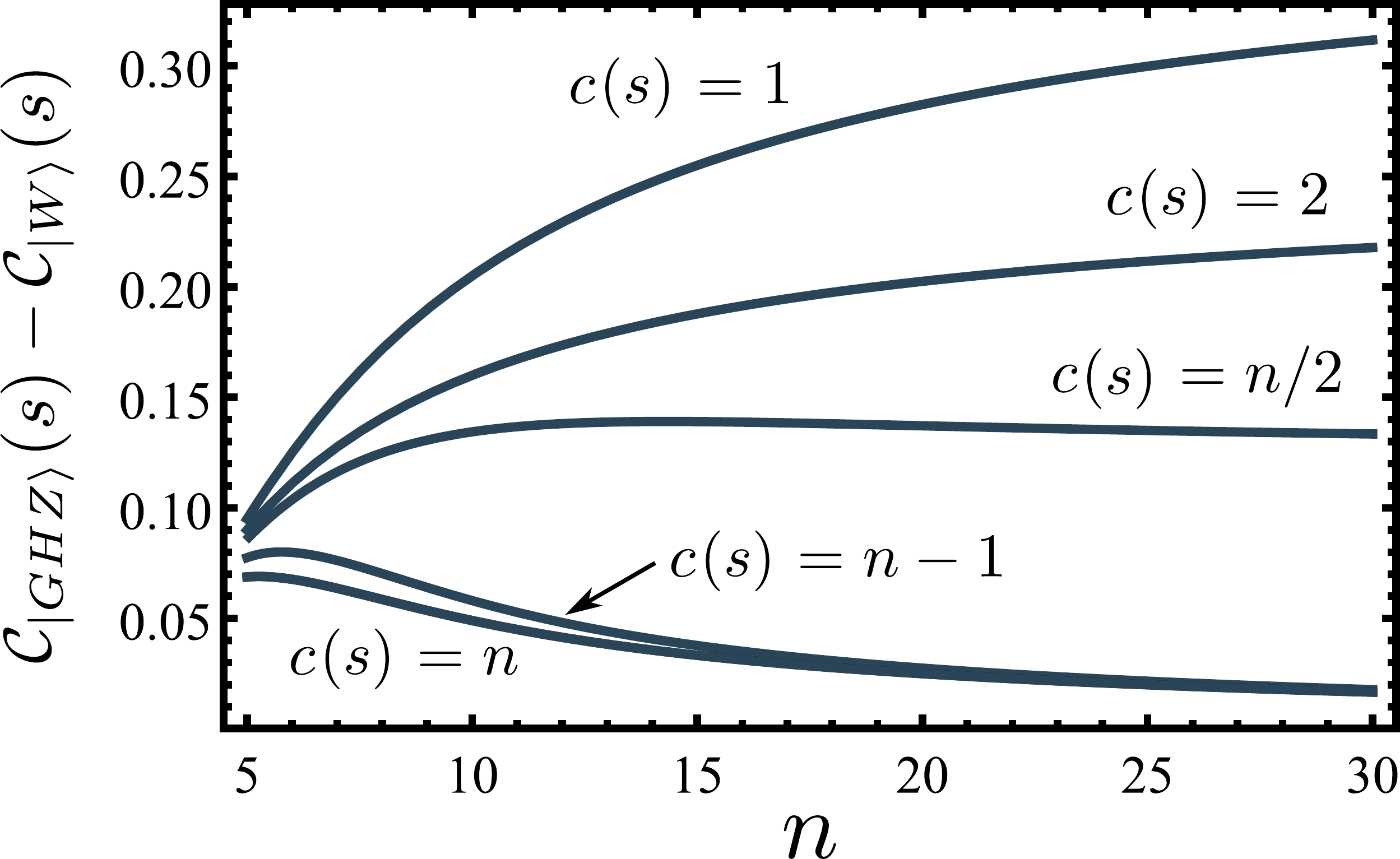}
\caption{\textbf{Comparison of the Concentratable Entanglements for $GHZ$ and $W$ states.} In the figure we show the difference  $\CC_{\ket{GHZ}}(s)-\CC_{\ket{W}}(s)$  versus the number of qubits $n$ for different sets $s$ with cardinalities $c(s)=1,2,n/2,n-1,n$. In all cases we find $\CC_{\ket{GHZ}}(s)>\CC_{\ket{W}}(s)$.   }
\label{fig:comparison}
\end{figure}

\textit{Examples.} Let us now showcase how the probabilities $p(\vec{z})$ and the Concentratable Entanglement can be used to characterize and quantify the multipartite entanglement in $n$-qubit $W$ and $GHZ$ states. First, let us consider the $W$-state $\ket{W}=\frac{1}{\sqrt{n}}\sum_{\vec{x}:w(\vec{x})=1}\ket{\vec{x}}$, i.e., the equal superposition of all states with Hamming weight equal to one. A direct calculation shows that $p(\vec{z})=\frac{1}{n^2}$ for all $\vec{z}$ with $w(\vec{z})=2$ and $p(\vec{z})=0$ for all $\vec{z}$ with $w(\vec{z})>2$. That is, in~\eqref{eq:concentratableven} one can only have terms where $\vec{z}$ has only two 1s. Concomitantly, when employing the $n$-qubit parallelized SWAP test one cannot concentrate the multipartite entanglement in $\ket{W}$ to simultaneously produce more than two Bell pairs. Then, noting that for a given $s$ there are $\sum_{\mu=1}^{c(s)}\binom{n-\mu}{1}=c(s)(2n-c(s)-1)/2$ non-zero terms in~\eqref{eq:concentratableven}, one finds  $\CC_{\ket{W}}(s)=c(s)(2n-c(s)-1)/2n^2$. 

On the other hand, consider the $GHZ$ state $\ket{GHZ}=\frac{1}{\sqrt{2}}\left(\ket{\vec{0}}+\ket{\vec{1}}\right)$. We now find  $p(\vec{z})=\frac{1}{2^n}$ for all $\vec{z}$ with (even) Hamming weight $w(\vec{z})\geq 2$. Unlike the $W$-state, when employing the $n$-qubit parallelized SWAP test one can obtain up to $n$ simultaneous Bell pairs. In this case,  given $s$, there are $\sum_{\mu=1}^{c(s)}\sum_{\nu=1}^{(n-\mu+1)/2}\binom{n-\mu}{2\nu-1}$ non-zero terms in~\eqref{eq:concentratableven}, leading to $\CC_{\ket{GHZ}}(s)=\frac{1}{2}\left(1-1/(2^{c(s)-\delta_{c(s)n}})\right)$,  
where the $\delta_{c(s)n}$ arises from the fact that $c(n)=n$ and $c(n)=n-1$ have the same number of terms. Note that, as expected, both  $\CC_{\ket{W}}(s)$ and $\CC_{\ket{GHZ}}(s)$ only depend on the cardinality of $s$ and not on the actual indices in the set, as both states are  invariant under permutations of the qubits.

We can now show that when  $s=\{j\}$ ($c(s)=1$), then $\CC_{\ket{W}}(\{j\})=\frac{n-1}{n^2}$ and $\CC_{\ket{GHZ}}(\{j\})=\frac{1}{4}$. This implies that  the bipartite entanglement of a single qubit in $\ket{W}$ decreases with $n$, while on the other hand  is constant for any $n$-qubit $\ket{GHZ}$ state. Moreover, if $s=\SC$ ($c(s)=n$), then $\CC_{\ket{W}}(\SC)=\frac{n-1}{2n}$ and $\CC_{\ket{GHZ}}(\SC)=\frac{1}{2}-\frac{1}{2^n}$, and we recover the results in~\cite{foulds2020controlled}. Note that for both cases considered one finds that $\CC_{\ket{GHZ}}(s)>\CC_{\ket{W}}(s)$, and hence that the Concentrable Entanglements detect more multipartite entanglement in $\ket{GHZ}$ than in $\ket{W}$. For $s=\SC$, however, in the limit of $n\rightarrow\infty$ both $\CC_{\ket{GHZ}}(s)$ and $\CC_{\ket{W}}(s)$ tend to the same value of $\frac{1}{2}$.  Lastly, we note that a property of $W$ states is that if one qubit is measured and projected out of the state, one can still measure entanglement in the ensuing state. Specifically, if one measures a qubit, then the  concentratable entanglement will be equal to  $c(s)(2n-c(s)-3)/2(n-1)^2$ (to $0$)  with probability $1-1/n^2$ ($1/n^2$), as this corresponds to  measuring the qubit in the zero (one) state. However, for $GHZ$ states, projecting out just one qubit always yields a state with zero concentratable entanglement -- confirming the well-known fact that while the $W$ state is less entangled than the $GHZ$, it is more robust to noise.

In Fig.~\ref{fig:comparison} we further analyze the difference $\Delta\CC=\CC_{\ket{GHZ}}(s)-\CC_{\ket{W}}(s)$ for different cardinalities of $s$.  Here we see that for  $c(s)\leq n/2$, $\Delta\CC$ increases (or remains constant) as $n$ increases implying that, again, small subsystems of qubits in $\ket{W}$ contain less multipartite entanglement than those in $\ket{GHZ}$. For $c(s)\sim n$, the difference  $\Delta\CC$ decreases as $n$ increases, showing that the total  multipartite entanglement is asymptotically the same for the two states.

\bigskip

\textit{Conclusion.} In this work, we introduced a computable and operationally meaningful family of entanglement monotones called the Concentratable Entanglements. For a pure state $\ket{\psi}$, these quantities can be estimated on a quantum computer given two copies of $\ket{\psi}$  via a parallelized SWAP test. We showed that they quantify and characterize the entanglement in and between subsystems of the composite quantum state in addition to quantifying global entanglement. We derived their operational meaning in terms of the probability of obtaining Bell pairs via the parallelized SWAP test. We also showed that well-known entanglement measures such as the $n$-tangle, Concurrence, and linear entropy of entanglement are recovered as special cases of Concentratable Entanglements. As a special case of our results, we proved a conjecture from Ref.~\cite{foulds2020controlled}, which claimed that the parallelized SWAP test could be used to quantify and categorize pure state multipartite entanglement.

An important future direction will be to experimentally observe our entanglement measures on real quantum devices (e.g., using a quantum optical Fredkin gate~\cite{milburn1989quantum}). A detailed analysis of the impact of hardware noise on the parallelized SWAP test will be useful for such implementations. As noise can turn pure states into mixed states, an additional important direction will be to generalize our entanglement measures (and their operational meaning) to mixed states.  Finally, one could also analyze if a swap test utilizing more copies of $\ket{\psi}$ could provide information that the Concentratable Entanglement cannot.

\section{Acknowledgments}

We thank Louis Schatzki for helpful discussions in deriving the continuity bound. We also thank Michael Walter, Felix Leditzky, and Greg Hamilton for their scrupulous read of the first version of this work, resulting in several small errors being corrected. This work was supported by the Quantum Science Center (QSC), a National Quantum Information Science Research Center of the U.S. Department of Energy (DOE).  JLB was initially supported by the U.S. DOE through a quantum computing program sponsored by the LANL Information Science \& Technology Institute. JLB was also supported by the National Science Foundation Graduate Research Fellowship under Grant No. 1650115. NG acknowledges support from the National Science Centre, Poland, under
the SONATA project ``Fundamental aspects of the quantum set of
correlations'', Grant No.~2019/35/D/ST2/02014). PJC also acknowledges initial support from the LANL ASC Beyond Moore's Law project. MC acknowledges initial support from the Center for Nonlinear Studies at Los Alamos National Laboratory (LANL). 

\bibliography{quantum.bib}

\clearpage
\newpage

\onecolumngrid
\setcounter{section}{0}
\setcounter{proposition}{0}
\setcounter{theorem}{0}
\setcounter{corollary}{0}
\setcounter{figure}{0}
\renewcommand{\figurename}{Sup. Fig.}

\section*{Supplementary Information for ``Computable and operationally meaningful multipartite entanglement measures''}

In this Supplemental Information, we provide additional details for the manuscript ``\textit{Computable and operationally meaningful multipartite entanglement measures}''. First, in Section~\ref{sec:proofs} we present the proofs and derivations of the main results of our manuscript. Then, in Section \ref{sec:error} we study the robustness of the Concentratable Entanglement against error in the copies that are input to the SWAP test. Finally, in Then, in Section \ref{sec:understanding} we present an in-depth exploration of the $n$-qubit parallelized SWAP test.

\section{Proofs of main results}\label{sec:proofs}

\subsection{Preliminaries}

Before proving all propositions from the main text, we state a few well-known facts and two useful lemmas.

First, let $S$ denote the SWAP operator acting on a bipartite Hilbert space $\mathcal{H} \otimes \mathcal{H}$, and let $\rho,\sigma \in \mathcal{H}$ be two quantum states. Then, it can be verified by direct computation that
\begin{align}\label{eq:swaptrick}
    \Tr\left[S (\rho \otimes \sigma)\right] = \Tr[\rho \sigma].
\end{align}
Note that by cyclicity of trace it is also true that $\Tr\left[(\rho \otimes \sigma)S\right] = \Tr[\rho \sigma]$. This is called the \textit{SWAP trick} and will be referred to as such throughout this Supplementary Information. For the special case where $\rho=\sigma$,  the SWAP trick gives $\Tr[S(\rho \otimes \rho)] = \Tr[\rho^2]$. Thus, given two copies of a quantum state, its purity can be estimated by measuring the SWAP operator. 

The second crucial fact is about the spectra of reduced pure quantum states. Let $\ket{\psi}$ be a pure state of a composite system $AB$ with associated density matrix denoted $\rho_{AB}=\dya{\psi}$. Then
\begin{align}
    \vec{\lambda}(\rho_A) &= \vec{\lambda}(\rho_B) 
\end{align}
where $\rho_A=\Tr_B[\rho_{AB}]$, $\rho_B=\Tr_A[\rho_{AB}]$, and $\vec{\lambda}(\rho)$ denotes the vector of non-zero eigenvalues of $\rho$ ordered in decreasing order. In other words, a reduced state of a composite system will have the same spectrum as its complement. This follows from the existence of the Schmidt decomposition for any bipartite pure state \cite{nielsen2000quantum}. 

Lastly, we note that the purity of a quantum state $\rho$, defined as $\Tr[\rho^2]$ is a convex function (i.e. $\Tr[ (p \rho_1 + (1-p)\rho_2) ^2] \leq p \Tr[\rho_1^2] + (1-p) \Tr[\rho_2^2]$ for $p\in [0,1]$). This can be seen in many ways, but we elect to show it from the non-negativity of the Hilbert-Schmidt distance ($0\leq D_{\text{HS}} (\rho_1,\rho_2)$). Recalling the definiton of the Hilbert-Schmidt distance as $D_{\text{HS}} (\rho_1,\rho_2) = \Tr[(\rho_1-\rho_2)^2]$, we can expand and write
\begin{align}
    0 &\leq \Tr[\rho_1^2] - 2 \Tr[\rho_1 \rho_2] +\Tr[\rho_2^2], \\
    2 \Tr[\rho_1 \rho_2]&\leq \Tr[\rho_1^2] + \Tr[\rho_2^2].
\end{align}
Then, multiplying through by $p(1-p)$, we can obtain
\begin{align}
    2 p(1-p) \Tr[\rho_1 \rho_2]&\leq p(1-p)  \Tr[\rho_1^2] + p(1-p) \Tr[\rho_2^2],\\
    &= (p-p^2)  \Tr[\rho_1^2] + \left[(1-p)-(1-p)^2\right] \Tr[\rho_2^2],\\
    p^2 \Tr[\rho_1^2] + 2p(1-p)\Tr[\rho_1\rho_2] +(1-p)^2 \Tr[\rho_2^2] &\leq p\Tr[\rho_1^2] + (1-p)\Tr[\rho_2^2],\\
    \Tr[ \left(p^2\rho_1^2+ 2p(1-p)\rho_1\rho_2 +(1-p)^2 \rho_2^2\right)] &\leq p\Tr[\rho_1^2] + (1-p)\Tr[\rho_2^2],\\
    \Tr[(p\rho_1 +(1-p)\rho_2)^2] &\leq p \Tr[\rho_1^2]+(1-p)\Tr[\rho_2^2],
\end{align}
as desired.

\bigskip

Finally, before proceeding to the proofs of the main results, we state the following Lemmas.

\begin{lemma}{\label{lemma:mono_local}}
Local purities (in any possible partition) cannot decrease, on average, under local operations.
\end{lemma}

\begin{proof}
Let ${\cal H}=\bigotimes_{k=1}^n {\cal H}_k$ be the joint Hilbert space of an $n$-partite system. Then, let $\Phi_k$ be a CPTP map   acting in a non-trivial way only on the $k$-th subsystem via the set of Kraus operators $\{M_{k}^{j_k}\}$. Note that these local operators satisfy the completeness relation $\sum_{j_k} M_{k}^{j_k \dag} M_{k}^{j_k}=\id_k$. As shown below, the action of $\Phi_k$ on a composite state $\rho$ will not decrease, on average, the purity of any reduced state, $\rho_p$, of $\rho$. 

As previously mentioned, let $\rho$ be an $n$-qubit quantum state in ${\cal H}$. The action of the Kraus operators on the $k$-th subsystem is  
\begin{align}\rho_{j_k}= \frac{1}{p_{j_k}} M_{k}^{j_k}\rho M_{k}^{j_k\dag},
\end{align}
with probability $p_{j_k}=\Tr[\rho M_{k}^{j_k\dag}M_{k}^{j_k}]$. Now let $\rho_p=\Tr_{p^c}[\rho]$ and $\rho_p^{j_k}=\Tr_{p^c}[\rho_{j_k}]$ be the reduced states of the $p$-th subsystem associated to $\rho$ and $\rho_{j_k}$ respectively. If $p\neq k$, then $\rho_p^{j_k}$ is the conditional state of $p$-th subsystem after an operation is performed on the $k$-th subsystem with outcome $j_k$, and these conditional states satisfy
\begin{equation}
    \rho_{p}=\sum_{j_k}p_{j_k}\rho_p^{j_k},
\end{equation}
which in turn implies the majorization relation (Theorem 11 in Ref. \cite{nielsen2001majorization})
\begin{equation}
    \vec{\lambda}(\rho_p)\prec \sum_{j_k} p_{j_k} \vec{\lambda}(\rho_p^{j_k})\label{eq:maj2},
\end{equation}
with $\vec{\lambda}(\sigma)$ the vector of eigenvalues of $\sigma$ arranged in decreasing order.
If, on the other hand, $p=k$, then the reduced state after outcome $j$ is obtained reads
\begin{equation}\label{eq:rho-k-jk}
    \rho_k^{j_k}=\frac{1}{p_{j_k}} \bar M_{k}^{j_k}\rho_{k} \bar M_{k}^{j_k \dag}\,,
\end{equation}
with $\bar M_{k}^{j_k}=\Tr_{k^c}[ M_{k}^{j_k}]$. From the polar decomposition of $\bar M_{k}^{j_k}\sqrt{\rho_k}$, it follows that there exists $V_{j_k}$ unitary such that \cite{nielsen1999conditions}
\begin{equation}\label{eq:polar-decomp}
        \sqrt{\rho_k}\bar M_{k}^{j_k \dag}\bar M_{k}^{j_k}\sqrt{\rho_k}= V_{j_k}\bar M_{k}^{j_k}\rho_{k} \bar M_{k}^{j_k \dag} V_{j_k}^\dag.
\end{equation}
Then, substituting Eq. \eqref{eq:rho-k-jk} into Eq. \eqref{eq:polar-decomp} and using the fact that $\sum_{j_k} (\bar{M}_k^{j_k}) ^\dag \bar{M}_k^{j_k} = \id_k $, one obtains
\begin{equation}\label{eq:rho-k}
    \rho_{k}=\sum_{j_k}p_{j_k} V_{j_k}\rho_k^{j_k}V_{j_k}^\dag.
\end{equation}
As above, Eq. \eqref{eq:rho-k} leads to the majorization relation
\begin{equation}
    \vec{\lambda}(\rho_k)\prec \sum_{j_k} p_{j_k}\vec{\lambda}(\rho_k^{j_k})\label{eq:maj}.
\end{equation}

Now, we use that fact that purity is a Schur-convex function (i.e. $f: \mathbb{R}^n \rightarrow \mathbb{R}$ such that $x \prec y \implies f(x) \leq f(y)$) \cite{nielsen1999conditions}. Thus, Eq. \eqref{eq:maj2} and Eq. \eqref{eq:maj} imply that 
\begin{align}
    \Tr[\rho_p^2] & \leq \Tr\left[\left(  \sum_{j_k} p_{j_k} \rho_p^{j_k} \right)^2\right]\\
    & \leq \sum_{j_k} p_{j_k} \Tr[(\rho_p^{j_k})^2],\label{eq:11}
\end{align}
where the final line holds due to the convexity of purity.  
 Thus we have shown that the action of $\Phi_k$ on $\rho$ will not decrease, on average, the purity of any of the conditional reduced states $\rho_p^{j_k}$ (relative to the original reduced state $\rho_p$).
\end{proof}

Lemma~\ref{lemma:mono_local} can be generalized as follows to separable operations.

\begin{lemma}{\label{lemma:mono_separablel}}
Local purities (in any possible partition) cannot decrease, on average, under separable operations.
\end{lemma}

\begin{proof}
Let $\Phi:{\cal H}\rightarrow{\cal H}$ be a separable CPTP map described by Kraus operators $\{{\cal M}_{\vec{j}}\}$ and let ${\cal H}=\bigotimes_{k=1}^n {\cal H}_k$ be the joint Hilbert space of an $n$-partite system. Because $\Phi$ is separable, it can be represented with Kraus operators of the form~\cite{horodecki2009quantum}
\begin{align}
    {\cal M}_{\vec{j}}=\bigotimes_{k=1}^n M_{k}^{j_k} = M_{1}^{j_1} \otimes M_{2}^{j_2} ... \otimes M_{n}^{j_n}, 
\end{align}
with $\vec{j}=(j_1,j_2,\ldots,j_n)$ and with the $\{M_{k}^{j_k}\}$ operators acting in a non-trivial way only on the $k$-th subsystem via the CPTP map $\Phi_k$.

As shown in Lemma~\ref{lemma:mono_local}, the action of each local operation $\Phi_k$ does not decrease, on average, the local purities. Then, noting that because majorization forms a partial order \cite{nielsen2001majorization}, it is transitive. That is, if $\vec{\lambda}(\rho_p)\prec  \sum_{j_k}p_{j_k}\vec{\lambda}(\rho_p^{j_k})$ and $\vec{\lambda}(\rho^{j_k}_p)\prec \sum_{j'_{k'}}p_{j'_{k'}}\vec{\lambda}(\rho^{j_k j'_{k'}}_p)$, then $\vec{\lambda}(\rho_p)\prec \sum_{j_k,j'_{k'}}p_{j_k}p_{j'_{k'}}\vec{\lambda}(\rho^{j_k j'_{k'}}_p)$. Here, $\rho^{j_k j'_{k'}}_p$ is the reduced state of the $p$-th subsystem after local operation are performed on the $k$-th and $k'$-th subsystem with outcomes $j_k$ and $j'_{k'}$, respectively. 

Hence, since the purity is a Schur-convex function, and the purity is a convex function, one finds
\begin{equation}\label{eq:purities_separable}
    \Tr[\rho_p^2] \leq \sum_{\vec{j}} p_{\vec{j}}  \Tr[(\rho^{\vec{j}}_p)^2]\,,
\end{equation}
where $ p_{\vec{j}}=\prod_{k=1}^n p_{j_k}$, such that $\sum_{\vec{j}} p_{\vec{j}}=1$, and where we employ the notation $\rho^{\vec{j}}_p=\rho^{j_1j_2\cdots j_n}_p$ for the reduced state on the $p$-th subsystem conditioned on the $\vec{j}$ outcome. Note that here we follow the same similar to that used in deriving Eq.~\eqref{eq:11}.  Moreover, we remark that the right hand side of~\eqref{eq:purities_separable} corresponds to the average purity of the conditional states obtained via the separable operation $\Phi$. Thus, we have shown that local purities cannot decrease, on average, under  separable operations $\Phi$, as desired.

\end{proof}

With these facts in mind, we can turn to proving our main results,  which we reproduce here to facilitate readability.

\subsection{Proof of Theorem 1} \label{sec:Thm1-proof}
\begin{theorem}
The Concentratable Entanglement has the following properties:
\begin{enumerate}
    \item $\CC_{\ket{\psi}}(s)$ is non-increasing, on average, under LOCC operations and hence is a well-defined pure state entanglement measure. 
    \item If $\ket{\psi}$ is a separable state of the form $\ket{\psi}=\bigotimes_{j=1}^n \ket{\phi_j}$, then  $\CC_{\ket{\psi}}(s)=0$ for all $s\in\PC(\SC)\setminus \{\emptyset\}$.
    \item $\CC_{\ket{\psi}}(s')\leq \CC_{\ket{\psi}}(s) $ if $s'\subseteq s$.  
    \item Subadditivity,   $\CC_{\ket{\psi}}(s\cup s')\leq \CC_{\ket{\psi}}(s) +\CC_{\ket{\psi}}(s') $ for $s\cap s'=\emptyset$.  
    \item Continuity, let $\ket{\psi}$ and $\ket{\phi}$ be two states such that $\Vert \dya{\psi}-\dya{\phi}\Vert_1 \leq \varepsilon$, then $|\CC_{\ket{\psi}}(s)-\CC_{\ket{\phi}}(s)|\leq 2 \varepsilon$.
\end{enumerate}
\end{theorem}

\begin{proof}
\begin{enumerate}
    \item Recall from Definition 1 in the main text, that the Concentratable Entanglement has the form
    \begin{align}
\CC_{\ket{\psi}}(s)=1-\frac{1}{2^{c(s)}}\sum_{\alpha\in \PC(s)} \Tr[ \rho_\alpha^2] \,.
\end{align}
By Lemma \ref{lemma:mono_separablel}, the purity of the reduced quantum states is non-decreasing, on average, under separable operations. Thus, we immediately have that the Concentratable Entanglement is non-increasing, on average, under LOCC (acting on pure states), because LOCC operations are a subset of separable operations~\cite{horodecki2009quantum,vidal1999entanglement}. 
\item For pure, separable input states, all possible reduced states $\rho_{\alpha}$ are also pure ($\Tr[\rho_{\alpha}^2]=1$) thus the Concentratable Entanglement becomes
\begin{align}
\CC_{\ket{\psi}}(s)&=1-\frac{1}{2^{c(s)}}\sum_{\alpha\in \PC(s)} \Tr[\rho_\alpha^2] \,, \\
&= 1-\frac{1}{2^{c(s)}}\sum_{\alpha\in \PC(s)} 1 \,, \\
&= 1- \frac{1}{2^{c(s)}} 2^{c(s)}\,,\\
\CC_{\ket{\psi}}(s)&=0\,,
\end{align}
as is required for a well-defined measure of entanglement \cite{vidal1999entanglement}.

\item As a consequence of Proposition \ref{prop:oddSI} (proved below in Section \ref{sec:prop3-proof} of this Supplementary Information), the Concentratable Entanglement can be expressed as 
\begin{align}
\CC_{\ket{\psi}}(s)=\sum_{\vec{z}\in \ZC^{\text{even}}_1(s)}p(\vec{z}),
\end{align}
where, as state in the main text, $\ZC_0(s)$ is defined as the set of all bitstring with $0$'s on all indices in $s$. Now, if $s' \subseteq s$, then $\ZC^{\text{even}}_1(s') \subseteq \ZC^{\text{even}}_1(s)$ Thus, we can write
\begin{align}
    \CC_{\ket{\psi}}(s)&=\sum_{\vec{z}\in \ZC^{\text{even}}_1(s)}p(\vec{z}),\\
    &=\sum_{\vec{z} \in \ZC^{\text{even}}_1(s) \setminus \ZC^{\text{even}}_1(s')} p(\vec{z}) + \sum_{\vec{z} \in \ZC^{\text{even}}_1(s')} p(\vec{z}), \\
    &\geq \sum_{\vec{z} \in \ZC^{\text{even}}_1(s')} p(\vec{z}), \\
    &=  \CC_{\ket{\psi}}(s')
\end{align}
as desired.
\item As above, we express the Concentratable Entanglement as 
\begin{align}
\CC_{\ket{\psi}}(s)=\sum_{\vec{z}\in \ZC^{\text{even}}_1(s)}p(\vec{z}),
\end{align}
where, as state in the main text, $\ZC_0(s)$ is defined as the set of all bitstring with $0$'s on all indices in $s$. Then we respectively define $\ZC^{\text{even}}_1(s)$ and $\ZC^{\text{odd}}_1(s)$ as the complements of $\ZC_0(s)$ with even and odd Hamming weight,  such that $\ZC_0(s)\cup \ZC^{\text{even}}_1(s)\cup \ZC^{\text{odd}}_1(s)=\ZC$. It follows that when $s \cap s' = \emptyset$, $\ZC^{\text{even}}_1(s \cup s') = \ZC^{\text{even}}_1(s) + \ZC^{\text{even}}_1(s') - \ZC^{\text{even}}_1(s) \cap \ZC^{\text{even}}_1(s)$. This allows us to write
\begin{align}
    C_{\ket{\psi}}(s\cup s') &= \sum_{\vec{z} \in \ZC^{\text{even}}_1(s \cup s') } p(\vec{z}), \\
    &= \sum_{\vec{z} \in \ZC^{\text{even}}_1(s)} p(\vec{z}) + \sum_{\vec{z} \in \ZC^{\text{even}}_1(s') } p(\vec{z}) -\sum_{\vec{z} \in \ZC^{\text{even}}_1(s) \cap \ZC^{\text{even}}_1(s') } p(\vec{z}), \\
    &\leq \sum_{\vec{z} \in \ZC^{\text{even}}_1(s)} p(\vec{z}) + \sum_{\vec{z} \in \ZC^{\text{even}}_1(s') } p(\vec{z}),\\
    &= C_{\ket{\psi}} (s) + C_{\ket{\psi}} (s'),
\end{align}
which is the desired subadditivity relation.

\item The proof for the continuity of the Concentratable Entanglement can be found in~\cite{schatzki2021entangled}. However, we here repeat the proof for consistency. Consider two quantum state $\ket{\psi}$ and $\ket{\phi}$, and let their density matrices respectively be $\rho$ and $\sigma$, and their reduced density matrices be $\rho_\alpha$ and $\sigma_\alpha$.  with reduced state respectively denoted as  First, note that the following chain of inqualities hold
\begin{align}
        |\Tr[{\sigma_\alpha^2-\rho_\alpha^2}]| & = |\Tr{[(\sigma_\alpha+\rho_s)(\sigma_\alpha-\rho_\alpha)]}|\\
                                   & \leq \norm{\sigma_\alpha+\rho_\alpha}_2\norm{\sigma_\alpha-\rho_\alpha}_2\\
                                   & \leq (\norm{\sigma_\alpha}_2+\norm{\rho_\alpha})\norm{\sigma_s-\rho_s}_2\\
                                   & \leq 2\norm{\sigma_\alpha-\rho_s}_2\\
                                   & \leq 2\norm{\sigma_\alpha-\rho_\alpha}_1\\
                                   & \leq 2\norm{\sigma-\rho}_1.
\end{align}
Where the first inequality comes from Cauchy-Schwarz, the second/third from triangle inequality (and the maximum H.S. norm for density matrices), the fourth from Schatten norm monotonicity, and the last from trace norm monotonicity under CPTP maps. Recall the definition of the Concentratable Entanglement and applying the bound above on the $2^n-2$ purities in the summation, we arrive at
    \begin{equation}
        |C(\ket{\psi})-C(\ket{\phi})| \leq 2\left(\frac{2^n-2}{2^n}\right) \left\Vert\dya{\phi}-\dya{\psi}\right\Vert_{1} \leq 2\left\Vert\dya{\phi}-\dya{\psi}\right\Vert_{1}
    \end{equation}

\end{enumerate}
\end{proof}

\subsection{Proof of Proposition 1}\label{sec:prop1-proof}
\begin{proposition}\label{prop:CE-probabSM}
The Concentratable Entanglement can be computed from the outcomes of the $n$-qubit parallelized SWAP test as 
\begin{align}
\CC_{\ket{\psi}}(s)=1-\sum_{\vec{z}\in \ZC_0(s)}p(\vec{z}),
\end{align}
where $\ZC_0(s)$ is the set of all bitstrings with $0$'s on all indices  in $s$. 
\end{proposition}

\begin{proof}

We can directly compute the outcome probabilities of the $n$-qubit parallelized SWAP test as follows. Consider a Hilbert space of the form $\cal H \otimes \cal H$. Let $\{\ket{i}\}$ be an orthonormal basis of $\cal H$, so that $B=\{\ket{i}\ket{i'}\}$ is an orthonormal product basis of $\cal H \otimes \cal H$. The Swap operator $S:\cal H \otimes \cal H \rightarrow \cal H \otimes \cal H$ is defined by its action on the elements of B:
\begin{equation}\label{eq:sop}
    S\,\ket{i}\ket{i'}=\ket{i'}\ket{i}\,\,\quad \forall\, \quad \ket{i}\ket{i'}\in B. 
\end{equation}
Note that $S$ is both unitary and Hermitian. From this definition, we can explicitly express the unitary transformation $U_s$ corresponding to the parallelized SWAP test depicted in Sup. Fig. \ref{fig:nSWAP}. Given an input state $\ket{\Psi}=\ket{\psi_1}\ket{\psi_2}\in \cal H \otimes \cal H$, we obtain an output state of the form
\begin{equation}\label{eq:clean-outputstate}
    U_{s}\ket{0}\ket{\Psi}=\frac{1}{2}\sum_{z=0}^{1} \ket{z}(\id+(-1)^z S)\ket{\Psi},
\end{equation}
where $\id$ denotes the identity operator on $\cal H \otimes \cal H$. If the control system is traced out, the evolution of the pair of input systems is given by a channel with Kraus operators $K_\pm=\frac{1}{2}(\id\pm S)$. These operators are projectors onto the two eigenspaces of $S$ with eigenvalues $\pm 1$.

Now let $\cal H$ have a tensor product structure itself, i.e., ${\cal H}=\bigotimes_{i=1}^n{\cal H}_i$. Further, let $B=\{\ket{\vec{i}} = \bigotimes_{k=1}^n \ket{i_k}, i_k=0, ..., {\dim({\cal H}_k)-1}\}$ be a basis of $\cal H$, so that the set $\{\ket{\vec{i}}\ket{\vec{i'}}\}$ is a basis of $\cal H \otimes \cal H$. If the parallelized SWAP test is performed pairwise on each of the subsystems, then the evolution of the pair of input systems is described by a channel with Kraus operators
\begin{equation}
    {\cal K}_{\vec{z}}=\frac{1}{2^n}\prod_k(\id_k+(-1)^{z_k}S_k),
\end{equation}
where $S_k: \mathcal{H}_k \otimes  \mathcal{H}_k \rightarrow \mathcal{H}_k \otimes  \mathcal{H}_k $ defined by their action on basis states, $S_k\,\ket{i_k} \ket{i'_k} = \ket{i'_k} \ket{i_k}$. The control register is an $n$-qubit system, so a measurement of the control register can yield and bitstring $\vec{z}\in \{0,1\}^n$. Specifically, after a von Neumann measurement described by projectors $\{\ketbra{\vec{z}} {\vec{z}}\}$ 
is performed on the control qubits, the probability of obtaining bitstring $\vec{z}$ is
\begin{eqnarray} \label{eq:pz}
    p(\vec{z})&=&\bramatket{\Psi}{\frac{1}{2^n}\prod_k(\id_k+(-1)^{z_k}S_k)}{\Psi}.
\end{eqnarray}

Using Eq. \eqref{eq:pz}, we can write
\begin{align}
    1-\sum_{\vec{z}\in \ZC_0(s)}p(\vec{z}) &= 1 - \sum_{\vec{z}\in \ZC_0(s)} \bramatket{\Psi}{\frac{1}{2^n}\prod_k(\id_k+(-1)^{z_k}S_k)}{\Psi},
\end{align}
where the sum is over the subset of bitstrings which have zeroes in the indices of $s$. That is, we sum over $z_k \in \{0,1\}$ for $k \notin s$ which is equivalent to a partial trace over the corresponding subsystem of dimension $2^{n-c(s)}$. Noting that for all remaining $k\in s$,  $z_k=0$, we obtain 
\begin{align}
    1-\sum_{\vec{z}\in \ZC_0(s)}p(\vec{z}) &= 1- \frac{1}{2^{c(s)}}\Tr\left[\rho_s\otimes\rho_s \prod_k(\id_k+S_k)\right], \\
    &= 1 - \frac{1}{2^{c(s)}} \sum_{\alpha \in \PC(s)} \Tr[\rho_{\alpha}^2], \\
    1-\sum_{\vec{z}\in \ZC_0(s)}p(\vec{z}) &=  C_{\ket{\psi}}(s), 
\end{align}
as desired.
\end{proof}

\subsection{Proof of Proposition 2}\label{sec:prop2-proof}
\begin{proposition}
Given the expansion of the state   $\ket{\psi}=\sum_{\vec{i}} c_{\vec{i}}\ket{i_1}\ket{i_2}\cdots\ket{i_n}$, the probabilities  $p(\vec{z})$ for any $\vec{z}\in \ZC$ are given by
\begin{equation}
    p(\vec{z})=\frac{1}{2^n}\sum_{\substack{\vec{i},\vec{i'},\vec{j},\vec{j'}}} c_{\vec{i}} c_{\vec{i'}} c^*_{\vec{j}} c^*_{\vec{j'}} T_{\vec{i}\vec{i'}\vec{j}\vec{j'}}(\vec{z})\,,
\end{equation}
where $T_{\vec{i}\vec{i'}\vec{j}\vec{j'}}(\vec{z})$ are given by $
    T_{\vec{i}\vec{i'}\vec{j}\vec{j'}}(\vec{z})=\prod_k(\delta_{i_k j_k}\delta_{i_k' j'_k}+(-1)^{z_k}\delta_{i_k j'_k}\delta_{i_k' j_k})$, and where $z_k$ denotes the $k$-th component of $\vec{z}$.
\end{proposition}

\begin{proof}

By expanding the  composite state $\ket{\Psi}$ in an orthonormal basis, $\ket{\Psi} = \sum_{\vec{i}\vec{i'}} c_{\vec{i}} c_{\vec{i'}} \ket{\vec{i}}\ket{\vec{i'}}$, the probability $p(\vec{z})$ of Eq.~\eqref{eq:pz2} reads
\begin{align}\label{eq:pz2}
     p(\vec{z})&=\frac{1}{2^n}\sum_{\vec{i},\vec{i'},\vec{j},\vec{j'}}c_{\vec{i}}c_{\vec{i'}} c^*_{\vec{j}}c^*_{\vec{j'}}\, T_{\vec{i}\vec{i'}\vec{j}\vec{j'}}(\vec{z}), 
\end{align}
where the coefficients $T_{\vec{i}\vec{i'}\vec{j}\vec{j'}}(\vec{z})$ are given by
\begin{eqnarray}
    T_{\vec{i}\vec{i'}\vec{j}\vec{j'}}(\vec{z})&=&\bra{\vec{j}}\bra{\vec{j'}}\prod_k(\id_k+(-1)^{z_k}S_k)\ket{\vec{i}}\ket{\vec{i'}},\notag\\
    &=&\prod_k \bra{j_k}\bra{j'_k}(\id_k+(-1)^{z_k}S_k)\ket{i_k}\ket{i'_k}, \notag\\
    &=&\prod_k(\delta_{i_k j_k}\delta_{i_k' j'_k}+(-1)^{z_k}\delta_{i_k j'_k}\delta_{i_k' j_k}).
\end{eqnarray}
Note that an alternative method of computing $p(\vec{z})$ is provided in Section \ref{sec:understanding-nqubits} of this Supplementary Information.
\end{proof}

\subsection{Proof of Proposition 3}\label{sec:prop3-proof}
\begin{proposition}\label{prop:oddSI}
If  $\vec{z}$ has odd Hamming weight (if $w(\vec{z})$ is odd) then $p(\vec{z})=0$.
\end{proposition}
\begin{proof}
Let $\vec{z}=z_1 \cdot z_2 \dotsm z_n$ be a bitstring of length $n$ with odd Hamming weight $w(z)$. Next, define $\SC_w = \{i \in \SC : z_i =1\}$. Finally, denote $c_{wx} = |\SC_w \cap x|$ for any set $x$. Now from Eq.~\eqref{eq:pz}, we know the probability of obtaining a bitstring $\vec{z}$ after measuring the control system can be expressed
\begin{align}
    p(\vec{z}) &= \frac{1}{2^n}\bra{\Psi} \prod_{k=1}^n (\id_k + (-1)^{z_k} S_k) \ket{\Psi}.
\end{align}
Expanding the product yeilds a summation of the form
\begin{align}
    p(\vec{z}) &= \frac{1}{2^n} \sum_{x \in \PC (\SC)} (-1)^{c_{wx}} \Tr[(\rho_x \otimes \rho_x) S_x],
\end{align}
where as in the main text, $\rho_{x}$ denotes the reduced state corresponding to input state $\ket{\psi}$. Using the fact that $\Tr[\rho_x \otimes \rho_x S_x] = \Tr[\rho_x^2]$, we can write
\begin{align}
    p(\vec{z}) &= \frac{1}{2^n} \sum_{x \in \PC (\SC)} (-1)^{c_{wx}} \Tr[\rho_x^2].
\end{align}
Now, we use the following two facts: 1) for pure states $\ket{\psi}$, an associated reduced state ($\rho_x$) and its complement ($\rho_{x^c}$) have the same spectra, and 2) because $w(\vec{z})$ is odd $c_{wx}$ and $c_{wx^c}$ will have opposite even/odd parity. Thus, the corresponding terms in the sum will have the same magnitude but opposite sign, leading to complete pairwise cancellation. To see the cancellation explicitly, we can intentionally double count (and insert a factor of $1/2$ to offset the double counting. Then, for $w(\vec{z})$ odd, we have
\begin{align}
    p(\vec{z}) &= \frac{1}{2^{n+1}} \sum_{x\in \PC(\SC)} (-1)^{c_{wx}} \Tr[\rho_x^2] + (-1)^{c_{wx}+1} \Tr[\rho_{x^c}^2],\\
   &= \frac{1}{2^{n+1}} \sum_{x\in \PC(\SC)} (-1)^{c_{wx}} \left( \Tr[\rho_x^2] -  \Tr[\rho_{x^c}^2]\right),\\
    &= 0.
\end{align}
\end{proof}
\subsection{Proof of Proposition 4}\label{sec:prop4-proof}
\begin{proposition}\label{prop:bi-separableSI}
Let $\ket{\psi}$ be a bi-separable state $\ket{\psi}=\ket{\psi}_A\otimes\ket{\psi}_B$. Then for any bitstring $\vec{z}$ of Hamming weight $w(\vec{z})=2$, where $z_k=z_{k'}=1$ we have  $p(\vec{z})=0$ if qubit $k$ is in subsystem $A$, and qubit $k'$ is in subsystem $B$.
\end{proposition}

\begin{proof}
Let $\SC_A = \{1,2,\dots,n_A\}$ and $\SC_B = \{n_A+1,n_A+2,\dots, n_A n_B\}$ denote the indices of the qubits in system $A$ and $B$ of a bi-separable state $\ket{\Psi} = \ket{\psi}_A \otimes \ket{\psi}_B$, respectively. Then, for indices of interest $i(j) \in S_{A(B)}$, define $\SC_{A(B)}^{i(j)} =\{i (j)\}$. Then, the probability of observing a bitstring with a $z_i = 1$ and $z_j=1$ yields summands of the form 
\begin{align}
    \pm \Tr[\rho^2_{x_A \cup x_B}] \mp \Tr[\rho^2_{x_A^c \cup x_B}] \mp \Tr[\rho^2_{x_A \cup x_B^c}] \pm \Tr[\rho^2_{x_A^c \cup x_B^c}]
\end{align}
where $x_{A(B)} \in \PC(\SC_{A(B)}^{i(j)})$, $x_{A(B)}^c \in \SC_{A(B)}^{i(j)} \setminus x_{A(B)} $, and where the sign of each term is determined by the cardinality of the intersection between $\SC_h$ and the current subscripted set. For example, the first term's sign is determined by $(-1)^{|\SC_h \cap (x_A \cup x_B)|}$. The second term will always have opposite parity as one can verify. Then, using the fact that for pure states $\vec{\lambda}(\rho_x)=\vec{\lambda}(\rho_{x^c})$, we can conclude that the first and second terms will cancel. The same argument applies to the third and fourth terms. All summands being zero guarantees that the sum is zero thus $p(z_{i,j}=1)=0$.
\end{proof}

\subsection{Proof of Proposition 5}\label{sec:prop5-proof}
\begin{proposition}
If $n$ is even, then $p(\vec{1})$ is an entanglement monotone as  we have 
\begin{equation}
    p(\vec{1})=\frac{ \tau_{(n)}}{2^n}\,,
\end{equation}
where $\tau_{(n)}$ is the $n$-tangle.
\end{proposition}

In Ref. \cite{jaeger2003entanglement}, the authors explore a measure of entanglement defined as
\begin{align}
    S^2_{(n)} = \Tr[\rho \tilde{\rho}], 
\end{align}
where $\tilde{\rho}:= (\sigma^2)^{\otimes n}\rho^{*} (\sigma^2)^{\otimes n}$ with $\rho^{*}$ the complex conjugate of $\rho$, and $\sigma^2$ the Pauli-$Y$ operator. The authors also show that this measure of entanglement is in fact equivalent to the $n$-tangle, defined as  $\tau_{(n)}=|\braket{\psi}{\tilde{\psi}}|^2$. To see this, observe that for pure state input, we have
\begin{align}
    S_{(n)}^2 &= \Tr[\rho \tilde{\rho}],\\
    &= \Tr[\ket{\psi}\braket{\psi}{\tilde{\psi}}\bra{\tilde{\psi}} ],\\
    &= \braket{\psi}{\tilde{\psi}}\braket{\tilde{\psi}}{\psi},\\
    &= |\braket{\psi}{\tilde{\psi}}|^2,\\
     S_{(n)}^2 &= \tau_{(n)}.
\end{align}
Further, when the input state is comprised of an \textit{even} number of qubits $n$,  they give the explicit expression
\begin{align}\label{eq:Sn}
    S^2_{(n)} = \frac{1}{2^n} \left((S_{0\dotsm 0})^2 - \sum_{k=1}^n \sum_{i_k =1}^3 (S_{0 \dotsm i_k \dotsm 0})^2 + \sum_{k,l=1}^n \sum_{i_k,i_l=1}^3 (S_{0\dotsm i_k \dotsm i_l \dotsm 0})^2 -  \dots + (-1)^n \sum_{i_1, \dots, i_n =1}^3 (S_{i_1 \dotsm i_n})^2 \right),
\end{align}
where $S_{i_1 \dotsm i_n} := \Tr[\rho \sigma^{i_1}_1 \otimes \dotsm \otimes \sigma^{i_n}_n]$ for $i_k \in \{0,1,2,3\}$ corresponding to the Pauli matrices. For example, $\sigma^2_k$  corresponds to  the Pauli-$Y$ operator on the $k$-th qubit. With these definitions in place, we are ready to prove the proposition.
\begin{proof}
From Eq. \eqref{eq:pz}, the probability of obtaining the all ones bitstring, denoted $\vec{1}$, when measuring the control qubits in the SWAP test is 
\begin{align}
   2^n p(\vec{1}) &= \bramatket{\Psi}{\prod_k(\id_k-S_k)}{\Psi},\\
\end{align}
where we emphasize that the identity and SWAP operators are acting on the $k$-th qubits of the test and copy states. As one can verify, the SWAP operator can be decomposed in the Pauli basis as $S_k = \frac{1}{2}(\id_k +\sigma_k^1+ \sigma_k^2 + \sigma_k^3)$. Using this identity, we can write
\begin{align}
    2^n p(\vec{1}) &= \bramatket{\Psi}{\prod_k \frac{1}{2}(\id_k-\sigma_k^1- \sigma_k^2 -\sigma_k^3)}{\Psi}.\\
\end{align}
Expanding this product yields
\small
\begin{align}
    2^n p(\vec{1}) &= \frac{1}{2^n} ( \bra{\Psi}I^{\otimes n}\ket{\Psi} - \sum_{k=1}^n \sum_{i_k=1}^3 \bra{\Psi} \sigma_{i_k} \ket{\Psi} + \sum_{k,l=1}^n \sum_{i_k,i_l =1}^3 \bra{\Psi}\sigma_{i_k}\otimes \sigma_{i_l}  \ket{\Psi} +\dotsm + (-1)^n \sum_{i_1, \dots, i_n=1}^3 \bra{\Psi} \sigma_{i_1} \otimes \dots \otimes \sigma_{i_n}\ket{\Psi}),
\end{align}
\normalsize
where $\sigma_{i_k}$ is implicit notation for $\sigma_{i}$ on the $k$-th qubit of the test \textit{and} copy state and identity everywhere else. Now, recalling that $\ket{\Psi}=\ket{\psi}\ket{\psi}$, we can write each term as a trace. For example, $\bra{\psi}\bra{\psi} \id^{\otimes n} \ket{\psi} \ket{\psi} = \Tr[\rho \id^{\otimes n}]^2$ where $\rho=\ket{\psi}\bra{\psi}$. We can then express the sum as 
\begin{align}
    2^n p(\vec{1}) &= \frac{1}{2^n} \Big((\Tr[ \rho \id^{\otimes n}])^2 - \sum_{k=1}^n \sum_{i_k=1}^3 (\Tr[\rho \sigma_{i_k}])^2, \\
    &+ \sum_{k,l=1}^n \sum_{i_k,i_l =1}^3 (\Tr[\rho \sigma_{i_k}\otimes \sigma_{i_l}])^2 - \dotsm + (-1)^n \sum_{i_1, \dots, i_n=1}^3 (\Tr[\rho \sigma_{i_1} \otimes \dotsm \otimes \sigma_{i_n}] )^2\Big).
\end{align}
Then, recalling $S_{i_1 \dotsm i_n} := \Tr[\rho \sigma^{i_1}_1 \otimes \dotsm \otimes \sigma^{i_n}_n]$ for $i_k \in \{0,1,2,3\}$, we can write
\begin{align}
    2^n p(\vec{1}) &= \frac{1}{2^n} \left((S_{0\dotsm 0})^2 - \sum_{k=1}^n \sum_{i_k =1}^3 (S_{0 \dotsm i_k \dotsm 0})^2 + \sum_{k,l=1}^n \sum_{i_k,i_l=1}^3 (S_{0\dotsm i_k \dotsm i_l \dotsm 0})^2 -  \dots + (-1)^n \sum_{i_1, \dots, i_n =1}^3 (S_{i_1 \dotsm i_n})^2 \right),\\
    2^n p(\vec{1})&= S_{(n)}^2.
\end{align}
Finally, we saw above that for pure state inputs, $S_{(n)}^2 = \tau_{(n)}$ and thus we have
\begin{align}
    p(\vec{1}) &= \frac{\tau_{(n)}}{2^n},
\end{align}
as we set out to show. This result immediately implies that $p(\vec{1})$ is itself an entanglement monotone.
\end{proof}

\subsection{Proof of Proposition 6}\label{sec:prop6-proof}
\begin{proposition}
Given two-copies of $\ket{\psi}$  if the $k$-th control qubit of the  $n$-qubit parallelized SWAP test  was measured in the state $\ket{1}$, then the joint post-measured state of the $k$-th qubits of each copy of $\ket{\psi}$ is the Bell state $\ket{\Phi^-}=\frac{1}{\sqrt{2}}\left(\ket{01}-\ket{10} \right)$.
\end{proposition}

\begin{proof}
Let us recall from Eq.~\eqref{eq:clean-outputstate}, that the output state from SWAP test is $\frac{1}{2}\sum_{z=0}^{1} \ket{z}(\id+(-1)^z S)\ket{\Psi}$. If the control system is traced out, the evolution of test and copy systems is governed by Kraus operators of the form $K_{\pm} = \frac{1}{2}(\id \pm S)$, which are simply projectors onto the eigenspaces of the SWAP operator with eigenvalues $\pm 1$. The control system being measured in the state $\ket{1}$, corresponds to the test and copy systems evolving according to the Kraus operator, $\frac{1}{2}(\id-S)$. Thus, the reduced state on the test and copy systems will be in the eigenstate of $\frac{1}{2}(\id-S)$, which as one can verify is
\begin{align}
    \ket{\Psi^-} = \frac{1}{\sqrt{2}} (\ket{01}-\ket{10}).
\end{align}
\end{proof}

This completes the section of proofs of our main results. The following section aims to provide a deeper understanding of the parallelized SWAP test for the interested reader.

\section{Robustness of Concentratable Entanglement under Unequal Input States}\label{sec:error}
In this section we show that the Concentratable Entanglement is robust against errors in the input states to the SWAP tests. While this was initially analyzed in ~\cite{foulds2020controlled} for $GHZ$ and $W$ states, we here derive a general  upper bound on the error that results from inputting unequal pure states to the SWAP test.
\subsection{Derivation of the Error Bound}
Consider inputting into the SWAP test general $n$-qubit pure states $\rho := \ket{\psi}\bra{\psi}$ and $\rho' = \ket{\psi'}\bra{\psi'}$ which satisfy:
\begin{align}\label{eq:lowdistance}
    D(\rho,\rho') < \varepsilon,
\end{align}
where $D(\rho,\rho') = \frac{1}{2} \Tr[|\rho-\rho'|]$ is the trace distance and where $\varepsilon > 0.$  Next, because the trace distance is non-increasing under CPTP maps \cite{nielsen2000quantum}, we have 
\begin{align} \label{eq:TrDistDataProcessing}
    D(\rho_{\alpha},\rho'_{\alpha}) \leq D(\rho,\rho'),
\end{align}
where, as in the main text, $\rho_{\alpha}$ is the reduced state obtained by taking the partial trace (which is a CPTP map) over all other parties in the full state $\rho$.
Furthermore, it has been shown that~\cite{coles2019strong}
\begin{align}\label{eq:HStracebound}
    \frac{1}{4} D_{\text{HS}}(\rho_{\alpha},\rho'_{\alpha}) \leq D^2(\rho_{\alpha},\rho'_{\alpha}),
\end{align}
where $D_{\text{HS}}(\rho_{\alpha}-\rho'_{\alpha}) = \Tr[(\rho_{\alpha}-\rho'_{\alpha})^2]$ is the Hilbert-Schmidt distance. With these inequalities in mind, we can write 
\begin{align}
    \varepsilon^2 &> D^2(\rho,\rho'),\\
    &\geq  D^2(\rho_{\alpha},\rho'_{\alpha}), \\
    &\geq  \frac{1}{4} D_{\text{HS}}(\rho_{\alpha},\rho'_{\alpha}),\\
    4\varepsilon^2 &\geq \Tr[(\rho_{\alpha}-\rho'_{\alpha})^2],\\
    &= \Tr[\rho^2_{\alpha}] +\Tr[\rho'^2_{\alpha}] -2\Tr[\rho_{\alpha}\rho'_{\alpha}].
\end{align}
Thus, we now have 
\begin{align}\label{eq:TrDistBound}
     0\leq \Tr[\rho^2_{\alpha}] +\Tr[\rho'^2_{\alpha}] -2\Tr[\rho_{\alpha}\rho'_{\alpha}] < 4\varepsilon^2 \,,
\end{align}
where we recall that the lower bound arises from the fact that the Hilbert-Schmidt distance is non-negative.
In the proof of Prop. \ref{sec:prop1-proof}, the input state was assumed to be of the form $\ket{\Psi}=\ket{\psi}\ket{\psi'}$, where $\ket{\psi'}$ need not equal $\ket{\psi'}$. Thus, the Concentratable Entanglement in the case where $\ket{\psi}\neq \ket{\psi'}$ is still defined as in Definition 1 in the main text. Denoting the Concentratable Entanglement when $\ket{\psi'}\neq \ket{\psi'}$ as $C_{\ket{\psi}\ket{\psi'}}$ we have from Definition 1 in the main text
\begin{align}
\CC_{\ket{\psi}\ket{\psi'}}(s)=1-\frac{1}{2^{c(s)}}\sum_{\alpha\in \PC(s)} \Tr[ \rho_\alpha \rho'_{\alpha}] \,.
\end{align}
We also note that the definition of Concentratable Entanglement can be rearranged to give
\begin{align}
\sum_{\alpha\in \PC(s)} \Tr[ \rho^2_\alpha] = 2^{c(s)} (1-\CC_{\ket{\psi}}(s))\,.
\end{align}
Thus, summing over Eq. \eqref{eq:TrDistBound}, we can use this rearranged form to obtain
\begin{align}
0\leq\sum_{\alpha\in \PC(s)}(\Tr[\rho^2_{\alpha}] +\Tr[\rho'^2_{\alpha}] -2\Tr[\rho_{\alpha}\rho'_{\alpha}]) &< \sum_{\alpha\in \PC(s)} 4\varepsilon^2, \\
0\leq2^{c(s)} (1-\CC_{\ket{\psi}}(s)) + 2^{c(s)} (1-\CC_{\ket{\psi'}}(s)) -2 \left(2^{c(s)} (1-\CC_{\ket{\psi}\ket{\psi'}}(s))\right)&< 2^{c(s)} 4\varepsilon^2,\\
0\leq \left[ \CC_{\ket{\psi}\ket{\psi'}}(s) - \CC_{\ket{\psi}}(s) \right] + \left[ \CC_{\ket{\psi}\ket{\psi'}}(s)- \CC_{\ket{\psi'}}(s) \right] &< 4\varepsilon^2. \label{eq:errorBound}
\end{align}
This result shows that total error in the Concentratable Entanglements due to assuming $\ket{\psi}=\ket{\psi'}$ when really $\ket{\psi}\neq\ket{\psi'}$ is small when $D(\rho,\rho')$ is small. That is, when the exact and error states are close in trace distance. 

\subsection{Numerical Demonstration of Error Bound}

\begin{figure}[t!]
\centering
\includegraphics[width=0.85\columnwidth]{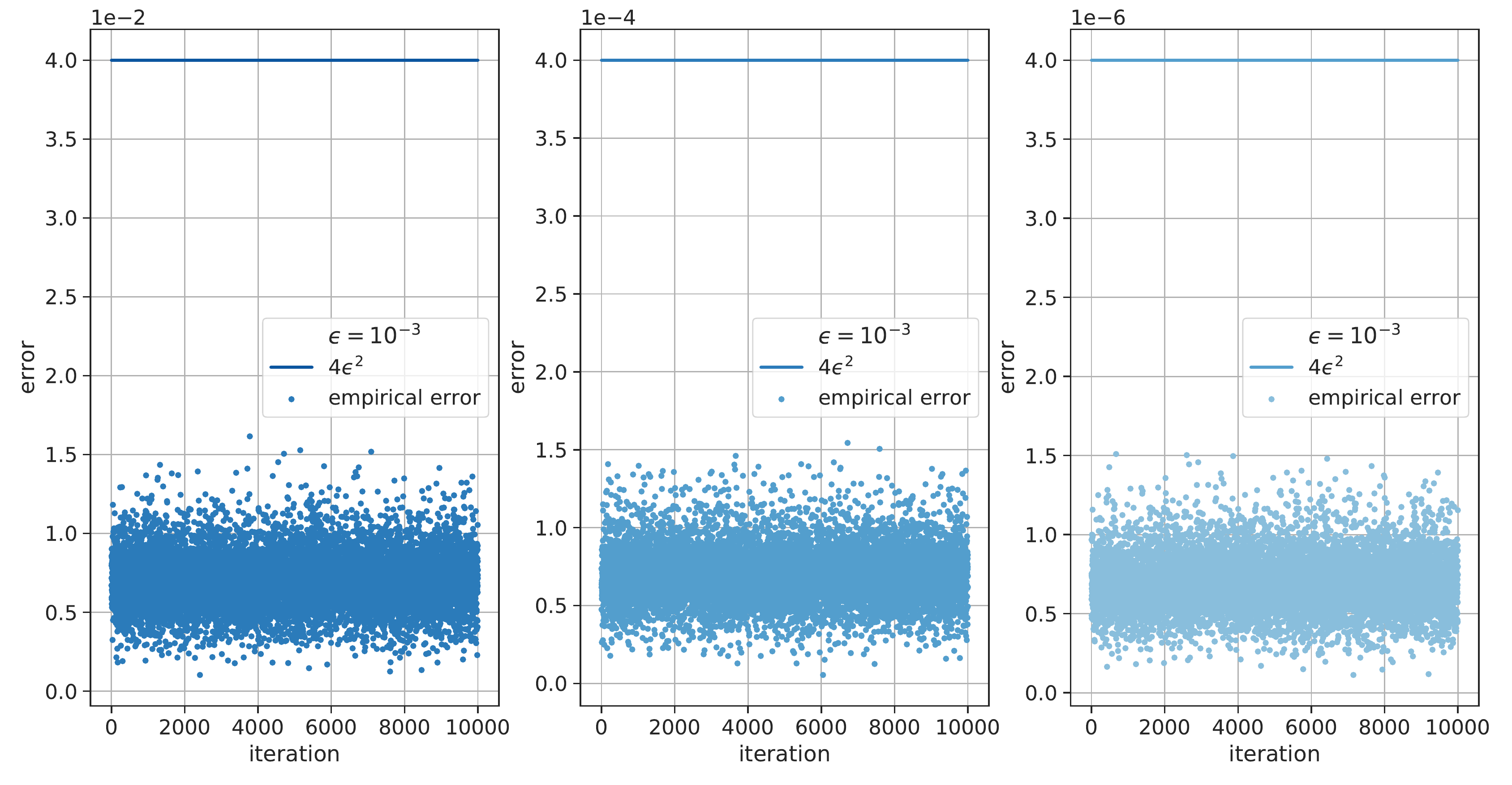}
\caption{\textbf{Verification of error bound for 3-qubit input states.} Here we  randomly generate 10,000 pairs of states according to the Haar measure which are $\varepsilon$-close in trace distance. We then plot Eq. \eqref{eq:errorBound} for errors ranging 3 orders of magnitude: $\varepsilon = 0.1, 0.001, 0.0001$.}
\label{fig:error}
\end{figure}

We now provide a 3-qubit numerical demonstration of the bound. First we need to generate random states that are $\varepsilon$-close in trace distance. Because we are considering pure states, we will make use of the fact that 
\begin{align}
D(\psi,\psi') = \frac{1}{2}\|\psi-\psi'\|_1=\sqrt{1-|\langle \psi | \psi' \rangle|^2}.
\end{align}
Then, we can generate a Haar random quantum state $|\psi\rangle$ and set $|\psi' \rangle = \delta |\psi\rangle + \sqrt{1-\delta^2} (\mathbb{I}-|\psi\rangle \langle \psi|)|0\rangle,$ where $\delta=\sqrt{1-\varepsilon^2}$. We can then verify that this construction results in random states that are $\varepsilon$-close in trace distance
\begin{align}
    \langle \psi | \psi' \rangle = \delta \langle \psi | \psi\rangle + \sqrt{1-\delta^2}(\langle \psi | 0\rangle - \langle \psi | 0 \rangle) = \delta.
\end{align}
Thus, the trace distance becomes
$$D(\psi,\psi')=\sqrt{1-\delta^2} = \varepsilon,$$
as desired. 

In Sup. Fig.~\ref{fig:error} we show the results obtained by sampling 10000 states with a given value of $\varepsilon$ and computing the difference in Concentratable Entanglement of Eq.~\eqref{eq:errorBound}.

\section{Understanding the Parallelized SWAP Test}\label{sec:understanding}
\subsection{Single Qubit}

\begin{figure}[h]
\includegraphics[width=0.9\columnwidth]{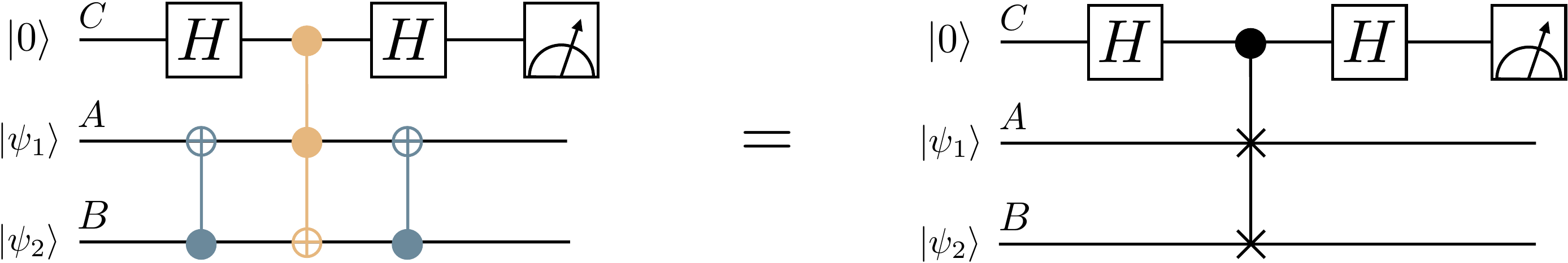}
\caption{\textbf{Single qubit SWAP test.} Given two single qubit states $\ket{\psi_1}$ and $\ket{\psi_2}$, the SWAP test is usually employed to compute the state overlap $|\langle \psi_1|\psi_2\rangle|^2$. The right circuit shows the decomposition of the parallelized SWAP gate into CNOT gates and a Toffoli gate. To simplify circuit diagrams, we will use the equivalent gate, called the Fredkin gate, as shown on the left.}
\label{fig:singleSWAP}
\end{figure}

Consider the SWAP test depicted in Sup. Fig~\ref{fig:singleSWAP}. Given two single-qubit pure states $\ket{\psi_1}$ and $\ket{\psi_2}$, one verify that the the SWAP test performs the following map 
\begin{align}
    \ket{0}\ket{\psi_1}\ket{\psi_2}\rightarrow &\frac{1}{\sqrt{2}}\left(\ket{0}+\ket{1}\right)\ket{\psi_1}\ket{\psi_2} \nonumber,\\
    \rightarrow&\frac{1}{\sqrt{2}}\left(\ket{0}\ket{\psi_1}\ket{\psi_2}+\ket{1}\ket{\psi_2}\ket{\psi_1}\right) \nonumber,\\
    \rightarrow&\frac{1}{2}\Big((\ket{0}+\ket{1})\ket{\psi_1}\ket{\psi_2}+(\ket{0}-\ket{1})\ket{\psi_2}\ket{\psi_1}\Big) \nonumber,\\
    \rightarrow&\frac{1}{2}\Big(\ket{0}(\ket{\psi_1}\ket{\psi_2}+\ket{\psi_2}\ket{\psi_1})\nonumber,\\
    &+\ket{1}(\ket{\psi_1}\ket{\psi_2}-\ket{\psi_2}\ket{\psi_1})\Big)\,.\label{eq:SWAPtest}
\end{align}
Hence, the probabilities $p(0)$ and $p(1)$ of respectively measuring the control qubit in the $\ket{0}$ and $\ket{1}$ state are given by
\begin{align}\label{eq:probability}
    p(0)=&\frac{1+|\langle\psi_1|\psi_2\rangle|^2}{2},\\
    p(1)=&\frac{1-|\langle\psi_1|\psi_2\rangle|^2}{2}\,.
\end{align}
Thus, we have $p(0)=1$ and $p(1)=0$ iff $\ket{\psi_1}=\ket{\psi_2}$ (up to a global phase), and conversely  one will measure the control qubit in the $\ket{1}$ state with maximum probability of $p(1)=1/2$ iff $\ket{\psi_1}$ and $\ket{\psi_2}$ are orthogonal.

\subsection{Two qubits}

To develop intuition for the inner workings of the  $n$-qubit SWAP test, we find it instructive to first examine the 2-qubit case. 
\begin{figure}[t]
\includegraphics[width=0.9\columnwidth]{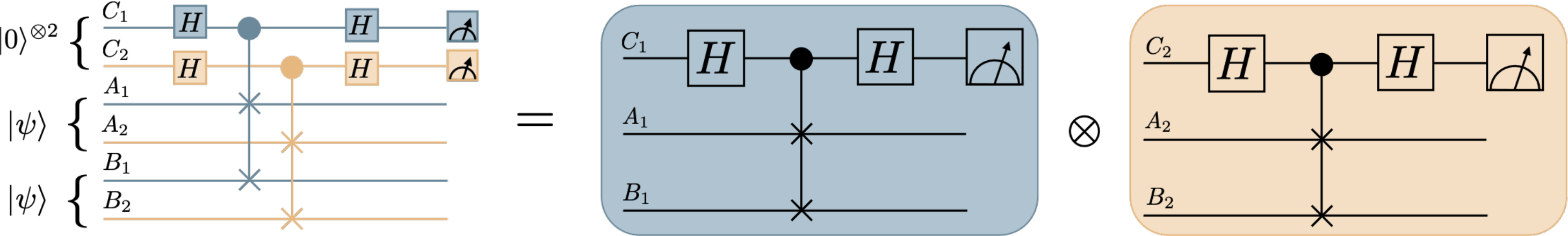}
\caption{\textbf{Two-qubit SWAP test.} The two-qubit SWAP test can be factored into the tensor product of two one-qubit SWAP tests.}
\label{fig:doubleSWAP}
\end{figure}
Let us first note that as shown in Sup. Fig.~\ref{fig:doubleSWAP}, the $i$-th control qubit is essentially performing a single qubit SWAP test between the $i$-th qubit of the test and copy states. That is, we can treat the $2$-qubit parallelized SWAP test as two, $1$-qubit parallelized SWAP tests. First, consider the separable input state $\ket{\Psi}=\ket{i_1}\ket{i_2}\ket{i'_1}\ket{i'_2}$ where we take $\{\ket{i}\}$ and $\{\ket{i'}\}$ to be computational basis vectors. Then, using Eq.~\eqref{eq:SWAPtest} we find that under the action of the two-qubit parallelized SWAP test, the total composite state

\begin{align}
    \ket{0}\ket{0}\ket{\Psi}=\ket{0}\ket{0}\ket{i_1}\ket{i_2}\ket{i'_1}\ket{i'_2}&= \Big(\ket{0}\ket{i_1}\ket{i'_1}\Big)\otimes\Big(\ket{0}\ket{i_2}\ket{i'_2}\Big)\nonumber
\,,\end{align}
becomes
\begin{align}
    \ket{\Psi_{\text{out}}}&=\frac{1}{2}\Big(\ket{0}(\ket{i_1}\ket{i'_1}+\ket{i'_1}\ket{i_1})+\ket{1}(\ket{i_1}\ket{i_1}-\ket{i'_1}\ket{i_1})\Big)\otimes\frac{1}{2}\Big(\ket{0}(\ket{i_2}\ket{i'_2}+\ket{i'_2}\ket{i_2})+\ket{1}(\ket{i_2}\ket{i'_2}-\ket{i'_2}\ket{i_2})\Big)\,.
\end{align}
Then, the probability of observing the control system in a particular state after a von Neumann measurement are given as
\begin{align}\label{eq:2qubitProductState}
\begin{split}
    p(00)&=\bra{0}\bra{0}\braket{\Psi}{\Psi_{\text{out}}} =\frac{1}{4}\left(1+\delta_{i_1,i_1'}\right)\left(1+\delta_{i_2,i_2'}\right), \quad
    p(01)=\bra{0}\bra{1}\braket{\Psi}{\Psi_{\text{out}}}=\frac{1}{4}\left(1+\delta_{i_1,i_1'}\right)\left(1-\delta_{i_2,i_2'}\right),\\
    p(10)&=\bra{1}\bra{0}\braket{\Psi}{\Psi_{\text{out}}}=\frac{1}{4}\left(1-\delta_{i_1,i_1'}\right)\left(1+\delta_{i_2,i_2'}\right), \quad
    p(11)=\bra{1}\bra{1}\braket{\Psi}{\Psi_{\text{out}}}=\frac{1}{4}\left(1-\delta_{i_1,i_1'}\right)\left(1-\delta_{i_2,i_2'}\right).
    \end{split}
\end{align}

Hence, for separable input states, if both copies of the 2-qubit input state are identical, we will measure the control to be in the $\ket{00}$ state with certainty. Now, consider a general 2-qubit state written in the computational product basis:  $\ket{\Psi}=\sum_{\vec{i}\vec{i'}} c_{\vec{i}}c_{\vec{i'}} \ket{i_1}\ket{i_2}\ket{i'_1}\ket{i'_2}$. We denote the conjugate of this state by $\bra{\Psi}=\sum_{\vec{j}\vec{j'}} c
^*_{\vec{j}}c^*_{\vec{j'}} \bra{j_1}\bra{j_2}\bra{j'_1}\bra{j'_2}$. In this notation, the resulting probability of observing the bitstring $\vec{z}=z_1 \cdot z_2$ when measuring the control register is
\begin{align}\label{eq:pz22}
   p(\vec{z}) = p(z_1 z_2) &= \sum_{\vec{i}\vec{i'}\vec{j}\vec{j'}}\frac{c_{\vec{i}} c_{\vec{i}'} c_{\vec{j}} c_{\vec{j}'}}{4} \left(\delta_{i_1,j_1}\delta_{i_1',j_1'}+(-1)^{z_1}\delta_{i_1,j_1'}\delta_{i_1',j_1}\right) \left(\delta_{i_2,j_2}\delta_{i_2',j_2'}+(-1)^{z_2}\delta_{i_2,j_2'}\delta_{i_2',j_2}\right), \\
   &=\sum_{\vec{i}\vec{i'}\vec{j}\vec{j'}}\frac{c_{\vec{i}} c_{\vec{i}'} c_{\vec{j}} c_{\vec{j}'}}{4} \prod_{k=1}^2 \left(\delta_{i_k,j_k}\delta_{i_k',j_k'}+(-1)^{z_k}\delta_{i_k,j_k'}\delta_{i_k',j_k}\right).
\end{align}
As we can see, this is actually the general expression for the outcome probabilities of the $n$-qubit parallelized SWAP test. The only difference is $k$ will range from $1$ to $n.$ Before turning to the general case, however, we should see how these quantities can be used to quantify entanglement. To this end, consider the two qubit Hilbert space $\HC=\HC_A\otimes\HC_B$, where $\HC_{A(B)}$ are two-dimensional Hilbert spaces. Then, any state $\ket{\psi}\in\HC$ has a \textit{Schmidt decomposition} given as
\begin{equation}\label{eq:schmidt2}
    \ket{\psi}=\sum_{j=1}^{r_s} \sqrt{\lambda_j} \ket{j_A}\ket{j_B}\,,
\end{equation}
where $\sqrt{\lambda_j}$ are known as the Schmidt coefficients, $r_s$ is the Schmidt rank, and where $\sum_j \lambda_j=1$. In addition, $\{\ket{j_{A(B)}}\}$ is a local basis for $\HC_{A(B)}$ called the Schmidt basis. Note that a state $\ket{\psi}\in \mathcal{H}$ is a product state if and only if it has a Schmidt rank of 1~\cite{nielsen2000quantum}. If $r_s=2$ then the state is entangled, and if $\lambda_1=\lambda_2=\frac{1}{2}$ then the state is maximally entangled. The Schmidt rank quantifies the amount of entanglement between systems $A$ and $B$ and is preserved under local unitary operations. Thus, the Schmidt rank allows us to define equivalence classes among bipartite entangled states. For example, all Bell states are equivalent under local unitary operations and thus form a unique class of maximally entangled states on two qubits \cite{bruss2019quantum}.

The, combining Eqs.~\eqref{eq:pz22} and~\eqref{eq:schmidt2} we find that the probability of measuring both control qubits in the $\ket{1}$ state is
\begin{equation}
    p(11)=\lambda_1\lambda_2=\frac{C^2}{4}\,,\label{eq:tq}
\end{equation}
where $C$ is the concurrence, a well-known measure of entanglement for $2$-qubit states \cite{wootters1998entanglement}.

\subsection{$n$-qubits} \label{sec:understanding-nqubits}
\begin{figure}[t]
\centering
\includegraphics[width=0.5\columnwidth]{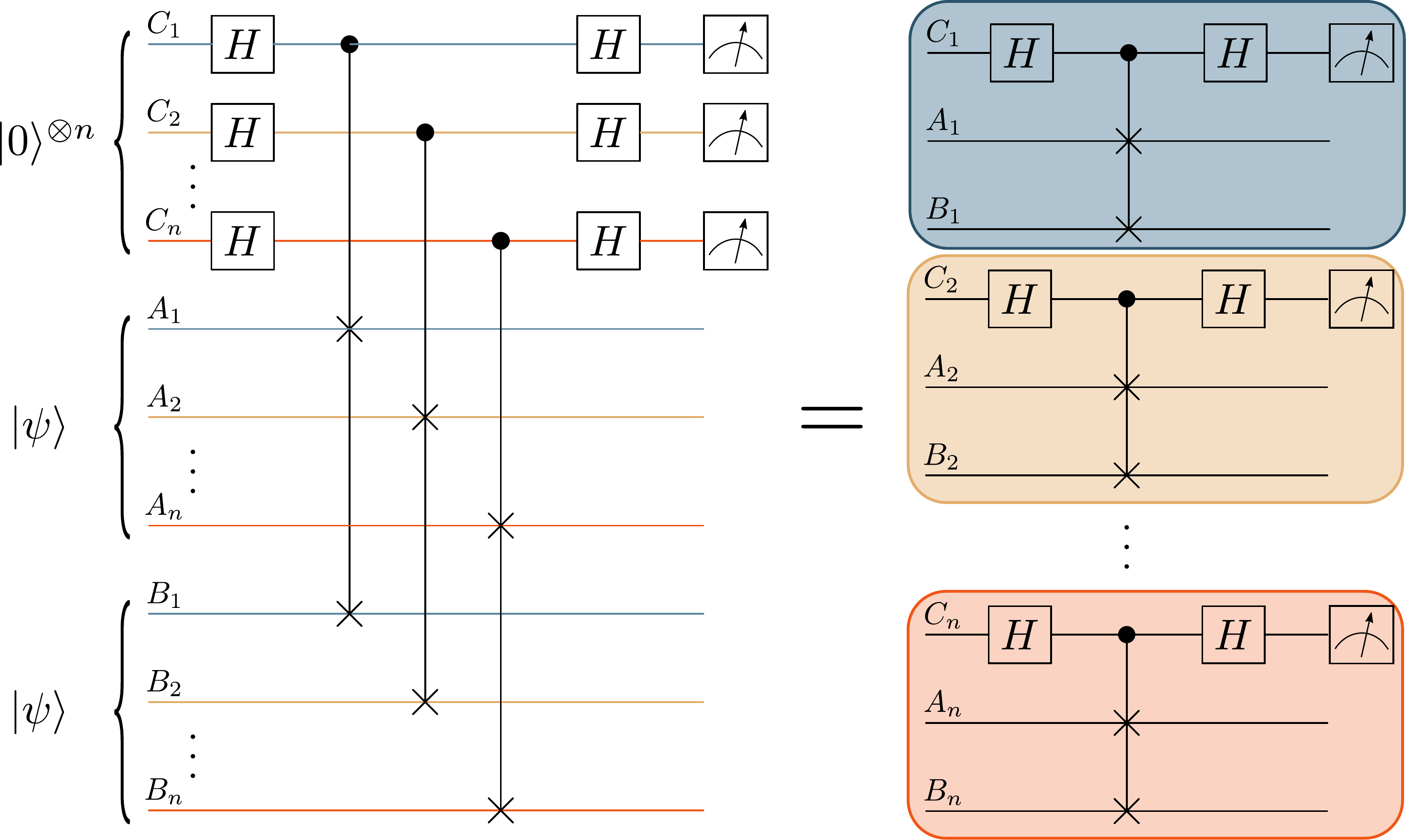}
\caption{\textbf{Parallelized $n$-qubit SWAP test.} a) Given two single qubit states $\ket{\psi_1}$ and $\ket{\psi_2}$, the SWAP test is usually employed to compute the state overlap $|\langle \psi_1|\psi_2\rangle|^2$. The addition of an arbitrary single-qubit rotation $U$ leaves the results of the SWAP test unchanged. }
\label{fig:nSWAP}
\end{figure}
Here we present an alternative derivation of Eq. \eqref{eq:pz}. Consider the three qubit product state formed as the tensor product of a test state $\ket{i}$, a copy state $\ket{i'}$ and a control qubit initialized in the zero state. We know that the controlled SWAP transformation yields
\begin{align}
    \ket{0}\ket{i}\ket{i'} \rightarrow \frac{1}{2} \left( \ket{0}(\ket{i}\ket{i'}+\ket{i'}\ket{i})+\ket{1}(\ket{i}\ket{i'}-\ket{i'}\ket{i})\right).
\end{align}
The $n$-qubit parallelized SWAP transformation can be decomposed into single-qubit SWAP tests on the $i$-th qubit of the test, copy, and control register. Thus we have
\begin{align}
    \ket{\vec{0}}\ket{\vec{i}}\ket{\vec{i'}} \rightarrow \frac{1}{2^n} \bigotimes_{k=1}^n \big( \ket{0}(\ket{i_k}\ket{i'_k}+\ket{i'_k}\ket{i_k})+\ket{1}(\ket{i_k}\ket{i'_k}-\ket{i'_k}\ket{i_k}) \big).
\end{align}
In general, the input state is of the form $\sum_{\vec{i},\vec{i'}} c_{\vec{i}} c_{\vec{i'}} \ket{\vec{i}}\ket{\vec{i'}}$ leading to 
\begin{align}
    \sum_{\vec{i},\vec{i'}} c_{\vec{i}} c_{\vec{i'}} \ket{\vec{0}}\ket{\vec{i}}\ket{\vec{i'}} \rightarrow \sum_{\vec{i},\vec{i'}} \frac{c_{\vec{i}} c_{\vec{i'}}}{2^n} \bigotimes_{k=1}^n \big( \ket{0}(\ket{i_k}\ket{i'_k}+\ket{i'_k}\ket{i_k})+\ket{1}(\ket{i_k}\ket{i'_k}-\ket{i'_k}\ket{i_k}) \big).
\end{align}
Letting $z_k \in \{0,1\}$, we can write the output state as
\begin{align}
   \ket{\Psi_{\text{out}}} &=  \sum_{\vec{i},\vec{i'}} \frac{c_{\vec{i}} c_{\vec{i'}}}{2^n} \bigotimes_{k=1}^n \sum_{\vec{z}}  \ket{z_k}\Big(\ket{i_k}\ket{i'_k}+ (-1)^{z_k}\ket{i'_k}\ket{i_k}\Big).
\end{align}
The corresponding density matrix is given as 
\begin{align}
    \rho_{\text{out}} = \sum_{\vec{i},\vec{i'},\vec{j},\vec{j'}} \frac{c_{\vec{i}} c_{\vec{i'}} c^*_{\vec{j}} c^*_{\vec{j'}}}{2^{2n}} \bigotimes_{k=1}^n \sum_{\vec{z},\vec{z'}}  \ket{z_k}\bra{z'_k}\Big(\ket{i_k}\ket{i'_k}+ (-1)^{z_k}\ket{i'_k}\ket{i_k}\Big)\Big(\bra{j_k}\bra{j'_k}+ (-1)^{z'_k}\bra{j'_k}\bra{j_k}\Big).
\end{align}
To obtain the probability of measuring bitstring $\vec{z}=z_1 \cdot z_2 \dotsm z_n$ we compute the expectation value of the operator $\ket{\vec{z}}\bra{\vec{z}}\otimes \id^{\otimes n} \otimes \id^{\otimes n}$ for the output state. We find
\begin{align}
    p(\vec{z}) &= \Tr[(\ket{\vec{z}}\bra{\vec{z}}\otimes \id^{\otimes n} \otimes \id^{\otimes n}) \rho_{\text{out}}],\\
    &= \sum_{\vec{i},\vec{i'},\vec{j},\vec{j'}} \frac{c_{\vec{i}} c_{\vec{i'}} c^*_{\vec{j}} c^*_{\vec{j'}}}{2^{2n}} \prod_{k=1}^n \Tr[\Big(\ket{i_k}\ket{i'_k}+ (-1)^{z_k}\ket{i'_k}\ket{i_k}\Big)\Big(\bra{j_k}\bra{j'_k}+ (-1)^{z_k}\bra{j'_k}\bra{j_k}\Big)],\\
     p(\vec{z}) &= \sum_{\vec{i},\vec{i'},\vec{j},\vec{j'}} \frac{c_{\vec{i}} c_{\vec{i'}} c^*_{\vec{j}} c^*_{\vec{j'}}}{2^n} \prod_{k=1}^n \left(\delta_{i_k j_k} \delta_{i'_k j'_k}+(-1)^{z_k} \delta_{i_k j'_k} \delta_{i'_k j_k} \right).
\end{align}

\end{document}